\documentclass[11pt,english]{article}
\usepackage[english]{babel}
\usepackage[margin=1in]{geometry}
\usepackage{hyperref}
\hypersetup{
     colorlinks=true,
     linkcolor=blue,
     citecolor=blue,
 }
\usepackage{bm,xfrac,amsmath,amsthm,amssymb,soul, comment, xspace}
\usepackage[T1]{fontenc} 
\usepackage{footnote}
\usepackage{nicefrac,bm}
\usepackage[ruled,vlined,linesnumbered]{algorithm2e}
\usepackage{thmtools}
\usepackage{thm-restate}
\usepackage{graphicx}
\usepackage{pifont,multirow}
\usepackage{textcomp}
\usepackage{lmodern}

\theoremstyle{plain}
\newtheorem{remark}{Remark}[section]
\newtheorem{example}{Example}[section]
\newtheorem{lemma}{Lemma}[section]
\newtheorem*{lemma*}{Lemma}

\newtheorem{theorem}{Theorem}[section]
\newtheorem{corollary}{Corollary}[section]

\newtheorem*{claim*}{Claim}

\newtheorem{definition}{Definition}[section]
\DeclareMathOperator*{\argmax}{\arg\!\max}
\usepackage{enumitem}
\newenvironment{case}{%
 \let\olditem\item% 
 \renewcommand\item[1][]{\olditem \textbf{##1} \\}%
 \begin{enumerate}[label=\textbf{Case \arabic*:},itemindent=*,leftmargin=0em]}{\end{enumerate}%
 }

\newcounter{experiment}[section]

\newcommand{\A}{{\mathcal A}}
\newcommand{\Goods}{{\mathcal G}}
\newcommand{\W}{{\mathcal W}}
\newcommand{\Set}{{\mathcal S}}

\newcommand{\M}{{\mathcal M}}
\newcommand{\Vals}{\mathcal V}
\newcommand{\NSWins}{(\A,\Goods,\Vals)}

\newcommand{\vecx}{{\mathbf x}}

\newcommand{\classP}{{\sf P}}
\newcommand{\classNP}{{\sf NP}}
\newcommand{\classAPX}{{\sf APX}}

\newcommand{\Maxcolor}{{\sf MAX\text{-}3\text{-}COLORING}}
\newcommand{\Alloc}{{\sf ALLOCATION}}

\newcommand{\EF}{\sf{EF1}}
\newcommand{\PO}{\sf{PO}}

\newcommand{\NSW}{\sf{NSW}}
\newcommand{\OPT}{\sf OPT}
\newcommand{\AlgSV}{\sf{RepReMatch}}
\newcommand{\AlgAdd}{\sf{SMatch}}

\newcommand{\AlgSVshort}{\sf{[R]}}
\newcommand{\AlgAddshort}{\sf{[S]}}

\newcommand{\SPLC}{\sf{SPLC}}
\newcommand{\BA}{\sf{BA}}
\newcommand{\GS}{\sf{GS}}
\newcommand{\OXS}{\sf{OXS}}
\newcommand{\XOS}{\sf{XOS}}

\let\oldnl\nl% Store \nl in \oldnl
\newcommand{\nonl}{\renewcommand{\nl}{\let\nl\oldnl}}% Remove line number for one line

\renewcommand{\nonl}{\renewcommand{\nl}{\let\nl\oldnl}}% Remove line number for one line
\long\def\symbolfootnote[#1]#2{\begingroup%
\def\thefootnote{\fnsymbol{footnote}}\footnote[#1]{#2}\endgroup}

\title{Approximating Nash Social Welfare under Submodular Valuations through (Un)Matchings\footnote{A preliminary version appeared in the Proceedings of the 31st Annual ACM-SIAM Symposium on Discrete Algorithms (SODA 2020)}}

\author{
Jugal Garg\thanks{University of Illinois at Urbana-Champaign. Supported by NSF CRII Award 1755619.}\\\texttt{jugal@illinois.edu} \and
Pooja Kulkarni\thanks{University of Illinois at Urbana-Champaign. Supported by NSF CRII Award 1755619.}\\ \texttt{poojark2@illinois.edu} \and
Rucha Kulkarni\thanks{University of Illinois at Urbana-Champaign. Supported by NSF CAREER Award CCF 1750436.}\\ \texttt{ruchark2@illinois.edu}
}
\begin{document}
\date{}
\maketitle
\thispagestyle{empty}
\begin{abstract}
We study the problem of approximating maximum Nash social welfare ($\NSW$) when allocating $m$ indivisible items among $n$ asymmetric agents with submodular valuations. The $\NSW$ is a well-established notion of fairness and efficiency, defined as the weighted geometric mean of agents' valuations. For special cases of the problem with symmetric agents and additive(-like) valuation functions, approximation algorithms have been designed using approaches customized for these specific settings, and they fail to extend to more general settings. Hence, no approximation algorithm with factor independent of $m$ is known either for asymmetric agents with additive valuations or for symmetric agents beyond additive(-like) valuations.

In this paper, we extend our understanding of the $\NSW$ problem to far more general settings. Our main contribution is two approximation algorithms for asymmetric agents with additive and submodular valuations respectively. Both algorithms are simple to understand and involve non-trivial modifications of a greedy repeated matchings approach. Allocations of high valued items are done separately by un-matching certain items and re-matching them, by processes that are different in both algorithms. We show that these approaches achieve approximation factors of $O(n)$ and $O(n \log n)$ for additive and submodular case respectively, which is independent of the number of items. For additive valuations, our algorithm outputs an allocation that also achieves the fairness property of envy-free up to one item ($\EF$).

Furthermore, we show that the $\NSW$ problem under submodular valuations is strictly harder than all currently known settings with an $\frac{\mathrm{e}}{\mathrm{e}-1}$ factor of the hardness of approximation, even for constantly many agents. For this case, we provide a different approximation algorithm that achieves a factor of $\frac{\mathrm{e}}{\mathrm{e}-1}$, hence resolving it completely.
\end{abstract}

\section{Introduction}\label{sec:intro}
We study the problem of approximating the maximum Nash social welfare ($\NSW$) when allocating a set $\Goods$ of $m$ indivisible items among a set $\A$ of $n$ agents with non-negative monotone \textit{submodular} valuations $v_i:2^{\Goods}\rightarrow \mathbb{R_+}$, and unequal or \textit{asymmetric} entitlements called \textit{agent weights}. Let $\Pi_n(\Goods)$ denote the set of all allocations, i.e., $\{(\vecx_1, \dots, \vecx_n)\ |\ \cup_i \vecx_i = \Goods;\  \vecx_i\cap \vecx_j = \emptyset, \forall i\neq j\}$. The $\NSW$ problem is to find an allocation maximizing the following weighted geometric mean of valuations,
\begin{equation}\label{eqn:1}
\argmax_{(\vecx_1, \ldots, \vecx_n)\in\Pi_n(\Goods)} \left(\prod\limits_{i\in\A} v_i(\vecx_i)^{\eta_i}\right)^{1/\sum_{i\in\A} \eta_i} \enspace ,
\end{equation}
where $\eta_i$ is the weight of agent $i$. We call this the \textit{Asymmetric Submodular $\NSW$ problem}.\footnote{In the rest of this paper, we refer to various special cases of the problem as the $\alpha$ $\mu$ $\NSW$ problem, where $\alpha$ is the nature of agents, symmetric or asymmetric, and $\mu$ is the type of agent valuation functions. We skip one or both qualifiers when they are clear from the context.} When agents are symmetric, $\eta_i=1, \forall i\in\A$. 

Fair and efficient allocation of resources is a central problem in economic theory. The $\NSW$ objective provides an interesting trade-off between the two extremal objectives of social welfare (i.e., sum of valuations) and max-min fairness, and in contrast to both it is invariant to individual scaling of each agent's valuations (see~\cite{Moulin03} for additional characteristics). It was independently discovered by three different communities as a solution of the bargaining problem in classic game theory~\cite{Nash50}, a well-studied notion of proportional fairness in networking~\cite{Kelly97}, and coincides with the celebrated notion of competitive equilibrium with equal incomes (CEEI) in economics~\cite{Varian74}. While Nash~\cite{Nash50} only considered the symmetric case,~\cite{HarsanyiS72,Kalai77} proposed the asymmetric case, which has also been extensively studied, and used in many different applications, e.g., bargaining theory~\cite{LV07,CM10,Thomson86}, water allocation~\cite{HoubaLZ14,DHYMQ18}, climate agreements~\cite{YIWZ17}, and many more.

The $\NSW$ problem is known to be notoriously hard, e.g., $\classNP$-hard even for two agents with identical additive valuations, and \classAPX-hard in general~\cite{Lee17}.\footnote{Observe that \emph{the partition problem} reduces to the $\NSW$ problem with two identical agents.} Efforts were then diverted to develop efficient approximation algorithms. A series of remarkable works~\cite{ColeG18,ColeDGJMVY17,AnariGSS17,AnariMGV18,BarmanKV18,GargHM18,CheungCGGHM18} provide good approximation guarantees for the special subclasses of this problem where agents are symmetric and have additive(-like) valuation functions\footnote{Slight generalizations of additive valuations are studied: budget additive~\cite{GargHM18}, separable piecewise linear concave (SPLC)~\cite{AnariMGV18}, and their combination~\cite{CheungCGGHM18}.} via utilizing ingenious different approaches. 
All these approaches exploit the symmetry of agents and the characteristics of additive-like valuation functions,\footnote{For instance, the notion of a maximum bang-per-buck (MBB) item is critically used in most of these approaches. There is no such equivalent notion for the submodular case.} which makes them hard to extend to the asymmetric case and more general valuation functions. As a consequence, no approximation algorithm with a factor independent of the number of items $m$~\cite{NguyenR14} is known either for asymmetric agents with additive valuations or for symmetric agents beyond additive(-like) valuations. These questions are also raised in~\cite{ColeDGJMVY17,BarmanKV18}.

The $\NSW$ objective also serves as a major focal point in fair division. For the case of symmetric agents with additive valuations, Caragiannis et al.~\cite{CaragiannisKMPSW16} present a compelling argument in favor of the `unreasonable' fairness of maximum $\NSW$ by showing that such an allocation has outstanding properties, namely, it is $\EF$ (a popular fairness property of envy-freeness up to one item) as well as Pareto optimal ($\PO$), a standard notion of economic efficiency. Even though computing a maximum $\NSW$ allocation is hard, its approximation recovers most of the fairness and efficiency guarantees; see e.g.,~\cite{BarmanKV18,CheungCGGHM18,GargM19}. 

In this paper, we extend our understanding of the $\NSW$ problem to far more general settings. Our main contribution is two approximation algorithms, $\AlgAdd$ and $\AlgSV$ for asymmetric agents with additive and submodular valuations respectively. Both algorithms are simple to understand and involve non-trivial modifications of a greedy repeated matchings approach. Allocations of high valued items are done separately by un-matching certain items and re-matching them, by processes that are different in both algorithms. We show that these approaches achieve approximation factors of $O(n)$ and $O(n \log n)$ for additive and submodular case respectively, which is independent of the number of items. For additive valuations, our algorithm outputs an allocation that is also $\EF$. 

\subsection{Model}\label{sec:preliminaries}
We formally define the valuation functions we consider in this paper, and their relations to other popular functions. For convenience, we also use $v_i(j)$ instead of $v_i(\{j\})$ to denote the valuation of agent $i$ for item $j$. 
\begin{enumerate}
    \item Additive: Given valuation $v_i(j)$ of each agent $i$ for every item $j$, the valuation for a set of items is the sum of the individual valuations. That is, $\forall \Set\subseteq \Goods, v_i(\Set)=\sum_{j\in \Set}v_i(j).$ In the restricted additive case, $v_i(j) = \{0, v_j\}, \forall i$. 
\item Budget additive ($\BA$): Every agent has an upper cap on the maximum valuation she can receive from any allocation. For any set of items, the agent's total valuation is the minimum value from the additive value of this set and the cap. i.e., $\forall \Set\subseteq \Goods, v_i(\Set)=\min\{\sum_{j\in \Set}v_i(j),c_i\},$ where $c_i$ denotes agent $i$'s cap. 
\item Separable piecewise linear concave ($\SPLC$): In this case, there are multiple copies of each item. The valuation of an agent is piecewise linear concave for each item, and it is additively separable across items. Let $v_i(j,k)$ denote the agent $i$'s value for receiving $k^{th}$ copy of item $j$. Concavity implies that $v_i(j,1) \ge v_i(j,2) \dots, \forall i, j$. The valuation of agent $i$ for a set $\Set$ of items, containing $l_j$ copies of items $j$, is $v_i(\Set) = \sum_{j}\sum_{k=1}^{l_j} v_i(j,k)$. 
    \item Monotone Submodular: Let $v_i(\Set_1\mid \Set_2)$ denote the marginal utility of agent $i$ for a set $\Set_1$ of items over set $\Set_2$, where $\Set_1,\Set_2 \subseteq \Goods$ and $\Set_1 \cap \Set_2 = \emptyset.$ Then, the valuation function of every agent is a monotonically non decreasing function $v_i:2^\Goods\rightarrow \mathbb{R_+}$ that satisfies the submodularity constraint that for all $i\in \A, h\in \Goods, \Set_1,\Set_2\subseteq \Goods,$
    $$v_i(h\mid \Set_1\cup \Set_2)\le v_i(h\mid \Set_1).$$
\end{enumerate}
Other popular valuation functions are $\OXS$, gross substitutes ($\GS$), $\XOS$ and subadditive~\cite{Nisan07}. These function classes are related as follows:
\begin{equation}\notag
\sf Additive \subsetneq 
\begin{array}{c}
\SPLC \subsetneq \OXS \\
\BA
\end{array}
\subsetneq \GS \subsetneq Submodular \subsetneq \XOS \subsetneq subadditive\enspace . 
\end{equation}

\subsection{Results}
Table~\ref{table:results} summarizes approximation guarantees of the algorithms $\AlgSV$ and $\AlgAdd$ under popular valuation functions, formally defined in Section~\ref{sec:preliminaries}. Here, the approximation guarantee of an algorithm is defined as $\alpha$ for an $\alpha \ge 1$, if it outputs an allocation whose (weighted) geometric mean is at least $1/\alpha$ times the maximum (optimal) geometric mean. All current best known results are also stated in the table for reference. 

\begin{table*}[!ht]
\begin{center}
\begin{tabular}{| c | c | c| c | c | }
  \hline			
  \multirow{2}{*}{\textbf{Valuations}} & \multicolumn{2}{c|}{\textbf{Symmetric Agents}} & \multicolumn{2}{c|}{\textbf{Asymmetric Agents}} \\ \cline{2-5}
   & \textbf{Hardness} & \textbf{Algorithm}  & \textbf{Hardness}& \textbf{Algorithm}\\ \hline
  Restricted Additive  & $1.069$ \cite{GargHM18} & $1.45$\cite{BarmanKV18} $\AlgAddshort$  &$1.069$ \cite{GargHM18}  &  $O(n)$ $\AlgAddshort$\\\hline
  Additive  & \textemdash \textquotedbl \textemdash &  $1.45$ \cite{BarmanKV18} & \textemdash \textquotedbl \textemdash & \textemdash \textquotedbl \textemdash \\\hline
  Budget additive  & \multirow{2}{*}{\textemdash \textquotedbl \textemdash} &  \multirow{2}{*}{$1.45$~\cite{CheungCGGHM18}} & \multirow{2}{*}{\textemdash \textquotedbl \textemdash} &  \multirow{2}{*}{\textemdash \textquotedbl \textemdash} \\
  $\SPLC$ & & & & \\\hline
  $\OXS$ & \multirow{2}{*}{\textemdash \textquotedbl \textemdash} & \multirow{2}{*}{$O(n\log n)$ $\AlgSVshort$} & \multirow{2}{*}{\textemdash \textquotedbl \textemdash}& \multirow{2}{*}{$O(n\log n)$ $\AlgSVshort$} \\
  Gross substitutes   &  &  &  & \\\hline
  Submodular &  $1.5819$ [Thm \ref{thm:submod_hardness}] & \textemdash \textquotedbl \textemdash & $1.5819$ [Thm \ref{thm:submod_hardness}] & \textemdash \textquotedbl \textemdash \\\hline
  $\XOS$ & \multirow{2}{*}{\textemdash \textquotedbl \textemdash} & \multirow{2}{*}{$O(m)$ \cite{NguyenR14}} & \multirow{2}{*}{\textemdash \textquotedbl \textemdash}& \multirow{2}{*}{$O(m)$ \cite{NguyenR14}} \\
  Subadditive &  &  &  & \\\hline
\end{tabular}\label{table:results}
\end{center}
\caption{Summary of results. Every entry has the best known approximation guarantee for the setting followed by the reference, from this paper or otherwise, that establishes it. Here, $\AlgAddshort$ and $\AlgSVshort$ respectively refer to Algorithms $\AlgAdd$ and $\AlgSV$.}
\end{table*}

To complement these results, we also provide a $\frac{\mathrm{e}}{\mathrm{e}-1}=1.5819$-factor hardness of approximation result for the submodular $\NSW$ problem in Section \ref{sec:hardness}. This hardness even applies to the case when the number of agents is constant. This shows that the general problem is strictly harder than the settings studied so far, for which $1.45$ factor approximation algorithms are known. 

For the special case of the submodular $\NSW$ problem where the number of agents is constant, we describe another algorithm with a \textit{matching $1.5819$ approximation factor} in Section \ref{sec:special_cases}, hence resolving this case completely. Finally in the same section, we show that for the symmetric additive $\NSW$ problem, the allocation of items returned by $\AlgAdd$ also satisfies $\EF$. Finally, a $1.45$-factor guarantee can be shown for the further special case of restricted additive valuations, by showing that the allocation returned by the algorithm in this case is $\PO$. This matches the current best known approximation factor for this case.

\subsection{Techniques}
We describe the techniques used in this work in a pedagogical manner. We start with a naive algorithm, and build progressively more sophisticated algorithms by fixing the main issues that result in bad approximation factors for the corresponding algorithms, finally ending with our algorithms.

All approaches compute, sometimes multiple, maximum weight matchings of weighted bipartite graphs. These graphs have agents and items in separate parts, and the edge weight assigned for matching an item $j$ to an agent $i$ is the logarithm of the valuation of the agent for the item, scaled by the agent's weight, i.e., $\eta_i \log{v_i(j)}$. Observe that, by taking the logarithm of the $\NSW$ objective \eqref{eqn:1}, we get an equivalent problem where the objective is to maximize the weighted sum of logarithms of agents' valuations. 

Let us first consider the additive $\NSW$ problem, and see what $\NSW$ is assured by computing a single such maximum weight matching. If the number of agents, say $n$, and items, say $m$, is the same, then the allocation obtained by matching items to agents according to such a matching results in the maximum $\NSW$ objective. ~\cite{NguyenR14} extend this algorithm to the general case, by allocating $n$ items according to one matching, and arbitrarily allocating the remaining items. They prove that this gives an $(m - n + 1)-$factor approximation algorithm. 

A natural extension to this algorithm is to compute more matchings instead of arbitrary allocations after a single matching. That is, compute one maximum weight matching, allocate items according to this matching, then repeat this process until all items are allocated. This repeated matching algorithm still does not help us get rid of the dependence on $m$ in the approximation factor. To see why, consider the following example.

\begin{example}\label{ex:rm_repmatch_counter}
Consider $2$ agents $A,B$ with weights $1$ each, and $m+1$ items. The valuations of $A$ and $B$ for the first item are $M+\epsilon$ and $M$ respectively. Agent $A$ values each of the remaining items at $1$, while $B$ only values the last of these at $1$, and values remaining $(m - 1)$ items at $0$. An allocation that optimizes the $\NSW$ of the agents allocates the first item to $B$, and allocate all remaining items to $A$. The optimal geometric mean is $(Mm)^{1/2}$. A repeated matching algorithm, in the first iteration, allocates the first item to $A$, and the last to $B$. No matching can now give non zero valuation to $B$. The maximum geometric mean that can be generated by such an algorithm is $((M+\epsilon+m-1)1)^{1/2}<\sqrt{M+m}$. Thus, using $M:=m$, the ratio of these two geometric means depends on $m$.
\end{example}

The above example shows the critical reason why a vanilla repeated matching algorithm may not work. In the initial matchings, the algorithm has no knowledge of how the agents value the entire set of items. Hence during these matchings it might allocate the high valued items to the wrong agents, thereby reducing the $\NSW$ by a large factor. To get around this problem, our algorithm needs to have some knowledge of an agent's valuation for the unallocated (low valued) set of items, while deciding how to allocate high valued items. It can then allocate the high valued items correctly with this foresight.

It turns out that there is a simple way to provide this foresight when the valuation functions are additive(-like). Effectively, we keep aside $O(n)$ high valued items of each agent, and assign the other items via a repeated matching algorithm. We then assign the items we had set aside to all agents via matchings that locally maximize the resulting $\NSW$ objective. The collective set of items put aside by all agents will have all the high valued items that required the foresight for correct allocation as a subset. Because these items are allocated after allocating the low valued items, this algorithm allocates the high valued items more smartly. In Section $\ref{sec:splc}$, we describe this algorithm, termed $\AlgAdd$, and show that it gives an $O(n)$ factor approximation for the $\NSW$ objective.

The above idea, however, does not work for submodular valuation functions. The main, subtle reason is as follows. Even in the additive case, the idea actually requires to keep aside not the set of items with highest valuation, but the set of items that leave a set of lowest valuation. For additive valuations, these sets are the same. However, it is known from \cite{SvitkinaF11} that finding a set of items of minimum valuation with lower bounded cardinality for monotone submodular functions is inapproximable within $\sqrt{m/\log m}$ factor, where $m$ is the number of items.  

We get around this issue and get the foresight for assigning high valued items in a different way. Interestingly, we use the technique of repeated matchings itself for this. In algorithm $\AlgSV$, we allocate items via repeated matchings, then release some of the initial matchings and re-match the items of these initial matchings. 

The idea is that the initial matchings will allocate all high valued items, even if incorrectly, and give us the set of items that must be allocated correctly. If the total number of all high valued items depends only on $n$, then the problem of maximizing the $\NSW$ objective when allocating this set of items is solved up to some factor of $n$ by applying a repeated matching algorithm. In Lemma $\ref{lem:numwinnerssub}$ we prove such an upper bound on the number of initial matchings to be released. 

Thus far, we have proved that we can allocate one set of items, the high valued items, approximately optimally. Now submodular valuations do not allow us to add valuations incurred in separate matchings to compute the total valuation of an agent. Getting such a repeated matchings type cumulative approach result in high total valuation requires the following natural modification in the approach. We redefine the edge weights used for computing matchings. We now consider the marginal valuations over items already allocated in previous matchings as edge weights rather than individual valuations. 

There are several challenges to prove this approach gives an allocation of high $\NSW$ overall. First, bounding the amount of valuation received by a particular agent as a fraction of her optimal allocation is difficult. This is because the subset of items allocated by the algorithm might be completely different from the set of optimal items. We can however give a relation between these two values and this is seen in Lemma $\ref{lem:main}.$  

Then, since we release and reallocate the items of initial matchings, now the set of items allocated to an agent can be completely different from the set before, changing all marginal utilities completely. It is thus non-trivial to combine the valuations from these stages too. This is done in the proof of Theorem $\ref{thm:mainsub}.$

Apart from this, we also have the following results in the paper that use different techniques. 

\noindent
\textbf{Submodular $\NSW$ with constant number of agents.} We completely resolve this case using a different approach that uses techniques of maximizing submodular functions over matroids developed in \cite{ChekuriVZ10} and a reduction from \cite{Vondrak08}. At a high level, we first maximize the continuous relaxations of agent valuation functions, then round them using a randomized algorithm to obtain an integral allocation of items. The two key results used in designing the algorithm are Theorems \ref{thm:multiple_submod} and \ref{thm:tail_bound_submod}.  

\noindent
\textbf{Hardness of approximation.} The submodular $\Alloc$ problem is to maximize the sum of valuations of agents over integral allocations of items. \cite{KhotLMM08} describe a reduction of $\Maxcolor$, which is $\classNP$-Hard to approximate within a constant factor, to $\Alloc$. We prove that this reduction also establishes the same hardness for the submodular $\NSW$ problem. 

\subsection{Further Related Work}
An extensive work has been done on special cases of the $\NSW$ problem. For the symmetric additive $\NSW$ problem, several constant-factor approximation algorithms have been obtained. The first such algorithm used an approach based on a variant of Fisher markets~\cite{ColeG18}, to achieve an approximation factor of $2\cdot e^{1/e} \approx 2.889$. Later, the analysis of this algorithm was improved to $2$~\cite{ColeDGJMVY17}. Another approach based on the theory of real stable polynomials gave an $e$-factor guarantee~\cite{AnariGSS17}. Recently, ~\cite{BarmanKV18} obtained the current best approximation factor of $e^{1/e} \approx 1.45$ using an approach based on approximate $\EF$ and $\PO$ allocation. These approaches have also been extended to provide constant-factor approximation algorithms for slight generalizations of additive valuations, namely the budget additive~\cite{GargHM18}, $\SPLC$~\cite{AnariMGV18}, and a common generalization of these two valuations~\cite{CheungCGGHM18}.

All these approaches exploit the symmetry of agents and the characteristics of additive-like valuation functions. For instance, the notion of a maximum bang-per-buck (MBB) item is critically used in most of these approaches. There is no such equivalent notion for the submodular case. This makes them hard to extend to the asymmetric case and to more general valuation functions.

Fair and efficient division of items to asymmetric agents with submodular valuations is an important problem, also raised in \cite{ColeDGJMVY17}. However, the only known result for this general problem is an $O(m)$-factor algorithm \cite{NguyenR14}, where $m$ is the number of items. 

Two other popular welfare objectives are the social welfare and max-min. In social welfare, the goal is to maximize the sum of valuations of all agents and in the max-min objective, the goal is to maximize the value of lowest-valued agent. The latter objective is also termed as the Santa Claus problem for the restricted additive valuations~\cite{BansalS06}.

The social welfare problem under submodular valuations has been completely resolved with a $\frac{e}{e-1} = 1.5819$-factor algorithm \cite{Vondrak08} and a matching hardness result~\cite{KhotLMM08}. Note that the additive case for this problem has a trivial linear time algorithm, hence it is perhaps unsurprising that a constant factor algorithm would exist for the submodular case. 

For the max-min objective, extensive work has been done on the restricted additive valuations case, resulting in constant factor algorithms for the same \cite{AnnamalaiKS15,DaviesRZ18}. However, for the unrestricted additive valuations the best approximation factor is $O(\sqrt{n}\log^3n)$~\cite{AsadpourS10}. For the submodular Santa Claus problem, there is an $O(n)$ factor algorithm \cite{KhotP07}. On the other hand, a hardness factor of $2$ is the best known lower bound for both settings \cite{BezakovaD05}. 

\noindent
\textbf{Organization of the paper:} In Section~\ref{sec:splc}, we describe the algorithm $\AlgAdd$ and analysis for the additive $\NSW$ problem. In Section \ref{sec:submod}, we present the algorithm $\AlgSV$ for submodular valuations. Section \ref{sec:hardness} contains the hardness proof for the submodular setting. The results for the special cases of submodular $\NSW$ with constant number of agents, symmetric additive $\NSW$, and symmetric additive $\NSW$ with restricted valuations are presented in Section \ref{sec:special_cases}. In Section \ref{sec:tightness}, we present counter examples to prove tightness of the analysis of Algorithms $\AlgSV$ and $\AlgAdd$. The final Section~\ref{sec:discussion} discusses possible further directions. 

\section{Additive Valuations}\label{sec:splc}
In this section, we present $\AlgAdd$, described in Algorithm \ref{NSWalgoADD}, for the asymmetric additive $\NSW$ problem, and prove the following approximation result.
\begin{theorem}\label{thm:additive}
The $\NSW$ objective of allocation $\vecx$, output by $\AlgAdd$ for asymmetric additive $\NSW$ problem, is at least $\nicefrac{1}{2n}$ times the optimal $\NSW$, denoted as $\OPT$, i.e., $\NSW(\vecx)\ge \frac{1}{2n}\OPT$.
\end{theorem}
$\AlgAdd$ is a single pass algorithm that allocates up to one item to every agent per iteration such that the $\NSW$ objective is locally maximized. An issue with a naive single pass, locally optimizing greedy approach is that the initial iterations work on highly limited information. As shown in Example $\ref{ex:rm_repmatch_counter}$, such algorithms can result in outcomes with very low $\NSW$ even for symmetric agents with additive valuation functions. In the example, although agent $A$ can be allocated an item of high valuation later, the algorithm does not \textit{know} this initially. Algorithm $\ref{NSWalgoADD}$ resolves this issue by pre-computing an approximate value that the agents will receive in later iterations, and uses this information in the edge weight definitions when allocating the first item to every agent. We now discuss the details of $\AlgAdd$. 

\subsection{Algorithm}
$\AlgAdd$ works in a single pass. For every agent, the algorithm first computes the value of $m-2n$ least valued items and stores this in $u_i$. $\AlgAdd$ then defines a weighted complete bipartite graph $\Gamma(\A, \Goods, \W)$ with edge weights $w(i,j)=\eta_i \log \left(v_i(j) + \frac{u_i}{n}\right),$ and allocates one item to each agent along the edges of a maximum weight matching of $\Gamma$. It then starts allocating items via repeated matchings. Until all items are allocated, $\AlgAdd$ iteratively defines graphs $\Gamma(\A, \Goods^{rem}, \W)$ with $\Goods^{rem}$ denoting the set of unallocated items and edge weights defined as $w(i,j)=\eta_i \log \left( v_i + v_i(j)\right)$, where $v_i$ is the valuation of agent $i$ for items that are allocated to her. $\AlgAdd$ then allocates at most one item to each agent according to a maximum weight matching of $\Gamma$. 
\begin{algorithm}[h!]
\label{NSWalgoADD}
\DontPrintSemicolon
\SetAlgoLined
\SetKwInOut{Input}{Input}
\SetKwInOut{Output}{Output}
\Input{A set $\A$ of $n$ agents with weights $\eta_i,\  \forall i\in\A$, a set $\Goods$ of $m$ indivisible items, and additive valuations $v_i:2^{\Goods}\rightarrow \mathbb{R}_+$, where $v_i(\Set)$ is the valuation of agent $i\in \A$ for a set of items $\Set\subseteq \Goods$.} 
\Output{An allocation that approximately optimizes the $\NSW$.}
\medskip
$\vecx_i \leftarrow \emptyset, u_i \leftarrow v_i(\Goods_{i, [2n+1:m]})\quad \forall i \in [n]$ \tcp*{$\Goods_{i,[a:b]}$ defined in Section \ref{subsec:SPLCnotation}}
Define weighted complete bipartite graph $\Gamma(\A,\Goods,\W)$ with weights $\W=\{w(i,j)\mid w(i,j)=\eta_i\log\left(v_i(j) + \frac{u_i}{n}\right),\forall i\in \A, j\in \Goods\}$ \;
Compute a maximum weight matching $\M$ for $\Gamma$\;
 $\vecx_i\leftarrow \vecx_i \cup \{j\mid (i,j) \in \M \},\ \forall i\in \A$\tcp*{allocate items according to $\M$}
 $\Goods^{rem}\leftarrow \Goods\backslash \{j \mid (i,j)\in \M \} $ \tcp*{update set of unallocated items} \medskip
\While{$\Goods^{rem} \neq \emptyset$}{
Define weighted complete bipartite graph $\Gamma(\A,\Goods^{rem},\W)$ with weights $\W=\{w(i,j)\mid w(i,j)=\eta_i\log(v_i(j) +v_i(\vecx_i)), \forall i\in \A, j\in \Goods^{rem}\}$ \;
 Compute a maximum weight matching $\M$ for $\Gamma$\;
 $\vecx_i\leftarrow \vecx_i \cup \{j\mid (i,j) \in \M \},\ \forall i\in \A$\tcp*{allocate items according to $\M$}
 $\Goods^{rem}\leftarrow \Goods^{rem}\backslash \{j \mid (i,j)\in \M \} $ \tcp*{remove allocated items}
}
Return $\vecx$
\caption{$\AlgAdd$ for the Asymmetric Additive $\NSW$ problem}
\end{algorithm}

\subsection{Notation}\label{subsec:SPLCnotation}
In the following discussion, we use $\vecx_i = \{h_i^1, \ldots, h_i^{\tau_i}\}$ to denote the set of $\tau_i$ items received by agent $i$ in $\AlgAdd$. We use $\vecx_i^* = \{g_i^1, \ldots, g_i^{\tau_i^*}\}$ to denote the set of $\tau_i^*$ items in $i$'s optimal bundle. Then for every $i$, all items in $\vecx_i$ and $\Goods$ are ranked according to the decreasing utilities as per $v_i$. We use the shorthand $[n]$ to denote the set $\{1, 2, \ldots, n\}$. Let $\Goods_{i,[a:b]}$ denotes the items ranked from $a$ to $b$ according to agent $i$ in $\Goods$, and $\vecx_{i,1:t}$ is the total allocation to agent $i$ from the first $t$ matching iterations. We also use $\Goods_{i, k}$ to denote the $k^{th}$ ranked item of agent $i$ from the entire set of items. For all $i$, we define $u_i$ as the minimum value for the remaining set of items upon removing at most $2n$ items from $\Goods$, i.e., $u_i = \min_{\Set \subseteq \Goods, |\Set| \leq 2n} v_i(\Goods \setminus \Set) = \Goods_{i, [2n+1, m]}$.\footnote{As the valuation functions are monotone, the minimum value will be obtained by removing exactly $2n$ items. The \emph{less than} accounts for the case when the number of items in $\Goods$ is fewer than $2n$.} 

\subsection{Analysis}\label{subsec:algSPLC}
To establish the guarantee of Theorem \ref{thm:additive}, we first prove a couple of lemmas.
\begin{lemma}\label{splc:lem:marginal_item}
$v_i(h_i^t)\geq v_i(\Goods_{i, tn}).$
\end{lemma}
\begin{proof}
Since every iteration of $\AlgAdd$ allocates at most $n$ items, at the start of iteration $t$ at most $(t-1)n$ items are allocated. Thus at least $n$ items from $\Goods$ ranked between $1$ to $tn$ by agent $i$ are still unallocated. In the $t^{th}$ iteration the agent will thus get an item with value at least the value of $\Goods_{i, tn}$ and the lemma follows.
\end{proof}

\begin{lemma}\label{lem:splc}
$v_i(h_i^2, \ldots, h_i^{\tau_i}) \geq \frac{u_i}{n}.$
\end{lemma}

\begin{proof}
Using Lemma $\ref{splc:lem:marginal_item}$ and since $v_i(\Goods_{i, tn}) \geq v_i(\Goods_{i, tn+k}), \ \forall k \in [n-1]$ 
$$v_i(h_i^t) \geq \frac{1}{n}v_i(\Goods_{i, [tn:(t+1)n - 1]})\enspace .$$
Thus,
$$v_i(h_i^2, \ldots, h_i^{\tau_i})=\sum_{t = 2}^{\tau_i} v_i(h_i^t)\geq \frac{1}{n}\sum_{t = 2}^{\tau_i} v_i(\Goods_{i,[tn:(t+1)n-1]})\enspace .$$
As at most $n$ items are allocated in every iteration, agent $i$ receives items for at least $\lfloor \frac{m}{n} \rfloor$ iterations.\footnote{Here we assume that the agents have non-zero valuation for every item. If it does not, the other case is also straightforward and the lemma continues to hold.} This implies that $(\tau_i + 1)n \geq m$ and hence,
\begin{align*}
     v_i(h_i^2, \ldots, h_i^{\tau_i}) &\geq \frac{1}{n} \big(v_i(\Goods_{i,[2n:m-1]})\big) \\
    &\geq \frac{1}{n} \big(v_i(\Goods_{i,[2n+1:m]})\big)  = \frac{1}{n} u_i\enspace .
\end{align*}
The second inequality follows as $v_i(\Goods_{i, 2n}) \geq v_i(\Goods_{i, m})$.
\end{proof}

We now prove the main theorem.
\begin{proof}[Proof of Theorem $\ref{thm:additive}$]
\begin{align*}
    \NSW(\vecx) &= \prod_{i = 1}^{n} \left( v_i(h_i^1, \ldots, h_i^{\tau_i})^{\eta_i} \right)^{\frac{1}{\sum_{i = 1}^{n} \eta_i }} \\
    &= \prod_{i = 1}^{n} \left( \left(v_i(h_i^1) + v_i(h_i^2 \ldots, h_i^{\tau_i}) \right)^{\eta_i} \right)^{\frac{1}{\sum_{i = 1}^{n} \eta_i }} \\
    &\geq \prod_{i = 1}^{n} \left( \left(v_i (h_i^1) + \frac{u_i}{n}\right)^{\eta_i} \right)^{\frac{1}{\sum_{i = 1}^{n} \eta_i }}, 
\end{align*}
where the last inequality follows from Lemma \ref{lem:splc}. During the allocation of the first item $h_i^1$, items $g_i^1$ of all agents are available. Thus, allocating each agent her own $g_i^1$ is a feasible first matching and we get
\begin{align*}
    \NSW(\vecx) &\geq \prod_{i = 1}^{n} \left( \left(v_i (g_i^1) + \frac{u_i}{n}\right)^{\eta_i} \right)^{\frac{1}{\sum_{i = 1}^{n} \eta_i }}. && \text{}
\end{align*}

Now, $u_i = \min_{\Set \subseteq \Goods, |\Set| \leq 2n} v_i(\Goods \setminus \Set)$. Suppose we define, $\Set_i^* = \arg \min_{|\Set| \leq 2n, \Set \subseteq \vecx_i^*} v_i(\vecx_i^* \setminus \Set)$, then $v_i(\vecx_i^* \setminus \Set_i^*) \leq u_i$. It follows by using $\Set_i = \arg\min_{\Set \subseteq \Goods, |\Set| \leq 2n} v_i(\Goods \setminus \Set)$, we get $u_i = v_i(\Goods \setminus \Set_i) \geq v_i(\vecx_i^* \setminus \Set_i) \geq v_i(\vecx_i^* \setminus \Set_i^*)$. Thus,
\begin{align*}    
    \NSW(\vecx) &\geq \prod_{i = 1}^{n} \left( \left(\frac{1}{2n}v_i (\Set_i^*) + \frac{1}{n}v_i(\vecx_i^* \setminus \Set_i^*)\right)^{\eta_i} \right)^{\frac{1}{\sum_{i} \eta_i }} \\
    &\geq \frac{1}{2n}\prod_{i = 1}^{n} \left( v_i(\vecx_i^*)^{\eta_i} \right)^{\frac{1}{\sum_{i = 1}^{n} \eta_i }} = \frac{1}{2n} \OPT\enspace .\qedhere
\end{align*}
\end{proof}

\begin{remark}\label{rem:AlgAddvsNaive}
When $\AlgAdd$ is applied to the instance of Example \ref{ex:rm_repmatch_counter}, it results in a better allocation than that of a naive repeated matching approach. Stage $1$ of $\AlgAdd$ computes $u_i$ as $m-2n$ and $0$ for A and B respectively. When this value is included in the edge weight of the first bipartite graph $\Gamma$, the resulting matching gives B the first item, and A some other item. Subsequently A gets all remaining items, resulting in an allocation having the optimal $\NSW$.  
\end{remark}

The algorithm $\AlgAdd$ easily extends to budget additive ($\BA)$ and separable piecewise concave ($\SPLC)$ valuations using the following small changes: In $\BA$, $u_i := \min(c_i, \Goods_{1, [2n+1:m]})$ where $c_i$ is the utility cap for agent $i$, and in $\SPLC$, $u_i$ needs to be calculated while considering each copy of an item as a separate item. In both cases, the edge weights in the bipartite graphs will use marginal utility (as we use in the submodular valuations case in Section~\ref{sec:submod}). Lemma \ref{lem:splc} and the subsequent proofs can be easily extended for these cases by combining ideas from Lemma \ref{lem:main} and Proof of Theorem \ref{thm:mainsub}. Thus, we obtain the following theorem. 

\begin{theorem}
The $\NSW$ objective of allocation $\vecx$, output by $\AlgAdd$ for asymmetric $\BA$ (and $\SPLC$) $\NSW$ problem, is at least $\nicefrac{1}{2n}$ times the optimal $\NSW$, denoted as $\OPT$, i.e., $\NSW(\vecx)\ge \frac{1}{2n}\OPT$.
\end{theorem}

\section{Submodular Valuations}\label{sec:submod}
In this section we present the $\AlgSV$, given in Algorithm~\ref{NSWalgoSubMod}, for approximating the $\NSW$ objective under submodular valuations. We will prove the following relation between the $\NSW$ of the allocation $\vecx$ returned by $\AlgSV$ and the optimal geometric mean $\OPT$.
\begin{theorem}\label{thm:mainsub}
The $\NSW$ objective of allocation $\vecx$, output by $\AlgSV$ for asymmetric submodular $\NSW$ problem, is at least $\nicefrac{1}{2n(\log{n}+2)}$ times the optimal $\NSW$, denoted as $\OPT$, i.e., $\NSW(\vecx)\ge \frac{1}{2n(\log{n}+2)}\OPT$.
\end{theorem}

\subsection{Algorithm }
$\AlgSV$ takes as input an instance of the $\NSW$ problem, denoted by $\NSWins$, where $\A$ is the set of agents, $\Goods$ is the set of items, and $\Vals=\{v_1,v_2\dotsc,v_n\}$ is the set of agents' monotone submodular valuation functions, and generates an allocation vector $\vecx$. Each agent $i \in \A$ is associated with a positive weight $\eta_i.$

$\AlgSV$ runs in three phases. In the first phase, in every iteration, we define a weighted complete bipartite graph $\Gamma(\A,\Goods^{rem},\W)$ as follows. $\Goods^{rem}$ is the set of items that are still unallocated ($\Goods^{rem}=\Goods$ initially). The weight of edge $(i,j), i\in \A, j\in \Goods^{rem}$, denoted by $w(i,j)\in \W$, is defined as the logarithm of the valuation of the agent for the singleton set having this item, scaled by the agent's weight. That is, $w(i,j)=\eta_i\log (v_i( j))$. We then compute a maximum weight matching in this graph, and allocate to agents the items they were matched to (if any). This process is repeated for $\lceil\log{n}\rceil$ iterations. 

We perform a similar repeated matching process in the second phase, with different edge weight definitions for the graphs $\Gamma$. We start this phase by assigning empty bundles to all agents. Here, the weight of an edge between agent $i$ and item $j$ is defined as the logarithm of the valuation of agent $i$ for the set of items currently allocated to her in Phase $2$ of $\AlgSV$, scaled by her weight. That is, if we denote the items allocated in $t$ iterations of Phase $2$ as $\vecx_{i, t}^2$, in $(t + 1)^{st}$ iteration, $w(i,j)=\eta_i \log (v_i(\vecx_{i, t}^{2}\cup \{j\})).$

In the final phase, we re-match the items allocated in Phase $1$. We release these items from their agents, and define $\Goods^{rem}$ as union of these items. We define $\Gamma$ by letting the edge weights reflect the total valuation of the agent upon receiving the corresponding item, i.e., $w(i,j)=\eta_i\log (v_i(\vecx_i^2\cup \{j\}))$, where $\vecx_i^2$ is the final set of items allocated to $i$ in Phase $2$. We compute one maximum weight matching for $\Gamma$ so defined, and allocate all items along the matched edges. All remaining items are then arbitrarily allocated. The final allocations to all agents, denoted as $\vecx=\{\vecx_i\}_{i\in \A}$, is the output of $\AlgSV$. 
\begin{algorithm}[h!]
\DontPrintSemicolon
\SetAlgoLined
\SetKwInOut{Input}{Input}
\SetKwInOut{Output}{Output}
\Input{A set $\A$ of $n$ agents with weights $\eta_i,\  \forall i\in\A$, a set $\Goods$ of $m$ indivisible items, and valuations $v_i:2^{\Goods}\rightarrow \mathbb{R}_+$, where $v_i(\Set)$ is the valuation of agent $i\in \A$ for a set of items $\Set\subseteq \Goods$.} 
\Output{An allocation that approximately optimizes the $\NSW$ objective}\medskip
\nonl \textbf{Phase $\bm 1$}: \medskip\;
$\vecx_i^1\leftarrow \emptyset,\  \forall i\in \A$ \tcp*{$\vecx_i^1$'s store the set of items allocated in Phase $1$}
$\mathcal{G}^{rem}\leftarrow \Goods$ \tcp*{set of unallocated items before every iteration}
$t\leftarrow 0$ \tcp*{iteration counter}
\While{$\Goods^{rem} \neq \emptyset$ and $t \leq \lceil\log{n}\rceil$}{
Define weighted complete bipartite graph $\Gamma(\A,\Goods^{rem},\W)$ with weights $\W=\{w(i,j)\mid w(i,j)=\eta_i\log(v_i(j)),\forall i\in \A, j\in \Goods^{rem}\}$ \;
 Compute a maximum weight matching $\M$ for $\Gamma$\;
 $\vecx_i^1 \leftarrow \vecx_i^1 \cup \{j\}, \quad \forall (i, j) \in \M$ \tcp*{allocate items to agents according to $\M$}
 $\Goods^{rem}\leftarrow \Goods^{rem}\backslash \{j \mid (i,j)\in \M \} ;\ t\leftarrow t+1$ \tcp*{remove allocated items}
}
\BlankLine
\nonl \textbf{Phase $\bm 2$}: \medskip\;
For all $i$, $\vecx_i^2 \leftarrow \emptyset$ \tcp*{$\vecx_i^2$'s are the sets of items allocated in Phase $2$}
\While{$\mathcal{G}^{rem} \neq \emptyset$}{
Define weighted complete bipartite graph $\Gamma(\A,\Goods^{rem},\W)$ with weights $W=\{w(i,j)\mid w(i,j)=\eta_i\log(v_i(\vecx_{i}^2 \cup \{ j \})),\forall i\in \A, j\in \Goods^{rem}\}$ \; 
Compute a maximum weight matching $\M$ for $\Gamma$\;
$\vecx_i^2 \leftarrow \vecx_i^2 \cup \{j\}, \quad \forall (i, j) \in \M$ \tcp*{allocate items to agents according to $\M$}
$\Goods^{rem}\leftarrow \Goods^{rem}\backslash \{j \mid (i,j)\in \M \}$ \tcp*{remove allocated items}
}
\BlankLine
\nonl \textbf{Phase $\bm3$}: \medskip\;
$\Goods^{rem}\leftarrow \bigcup_i \vecx_i^1$ \tcp*{release items allocated in Phase 1}
Define weighted complete bipartite graph $\Gamma(\A,\Goods^{rem},\W)$ with $W=\{w(i,j)\mid w(i,j)=\eta_i\log(v_i(\vecx_i^2\cup \{j\})),\forall i\in \A, j\in \Goods^{rem}\}$\; 
Compute a maximum weight matching $\M$ for $\Gamma$\;
$\vecx_i^2 \leftarrow \vecx_i^2 \cup \{j\}, \quad \forall (i, j) \in \M$\tcp*{allocate items to agents according to $\M$}
Arbitrarily allocate rest of the items to agents, let $\vecx=\{\vecx_i\}_{i\in\A}$ denote the final allocation \;
\Return{$\vecx$}
\caption{$\AlgSV$ for the Asymmetric Submodular $\NSW$ problem}\label{NSWalgoSubMod}
\end{algorithm}
\subsection{Notation }\label{subsec:SVnotation}

There are three phases in $\AlgSV$. We denote the set of items received by agent $i$ in Phase $p\in \{1,2,3\}$ by $\vecx_i^{p}$, and its size $|\vecx_i^{p}|$ by $\tau_i^p$. Similarly, $\vecx_i$ and $\tau_i$ respectively denote the final set of items received by agent $i$ and the size of this set. Note that Phase $3$ releases and re-allocates selected items of Phase $1$, thus $\tau_i$ is not equal to $ \tau_i^1 + \tau_i^2 + \tau_i^3.$ The items allocated to the agents in Phase $2$ are denoted by $\vecx_i^2 = \{ h_i^1,h_i^2\dotsc, h_i^{\tau_i^2} \}$. We also refer to the complete set of items received in iterations $1$ to $t$ of Phase $p$ by $\vecx_{i, t}^{p},$ for any $p\in\{1,2,3\}.$

For the analysis, the marginal utility of an agent $i$ for an item $j$ over a set of items $\Set$ is denoted by $v_i(j\mid \Set)=v_i(\{j\}\cup \Set)-v_i(\Set)$. Similarly, we denote by $v_i(\Set_1\mid \Set_2)$ the marginal utility of set $\Set_1$ of items over set $\Set_2$ where $\Set_1,\Set_2 \subseteq \Goods$ and $\Set_1 \cap \Set_2 = \emptyset.$ We use $\vecx^*=\{\vecx_i^*\mid i\in \A\}$ to denote the optimal allocation of all items that maximizes the $\NSW$, and $\tau_i^*$ for $|\vecx_i^*|$. For every agent $i$, items in $\vecx_i^*$ are ranked so that $g_i^j$ is the item that gives $i$ the highest marginal utility over all higher ranked items. That is, for $j=1$, $g_i^1$ is the item that gives $i$ the highest marginal utility over $\emptyset,$ and for all $2 \leq j \leq \tau_i^*,$ $g_i^j = \argmax_{g\in \vecx_i^* \backslash \{g_i^1, \dotsc, g_i^{j - 1} \}} v_i(g\mid \{ g_i^1, \ldots, g_i^{j - 1}\})$.\footnote{Since the valuations are monotone submodular, this ensures that $v_i(g_i^j \mid \{ g_i^1, \ldots, g_i^{j - 1}\}) \geq v_i(g_i^k \mid \{ g_i^1, \ldots, g_i^{k - 1}\})$ for all $ k \geq j.$ This implies that in any subset of $\ell$ items in the optimal bundle, the highest ranked item's marginal contribution is at least $1/\ell$ times that of this set, when the marginal contribution is counted in this way.} 

We let $\Bar{\vecx}_i^*$ denote the set of items from $\vecx_i^*$ that are not allocated (to any agent) at the end of Phase $1$, and we denote by $\Bar{v}_i^*=v_i(\Bar{\vecx}_i^*)$ and $\Bar{\tau}_i^*=|\Bar{\vecx}_i^*|$ respectively the total valuation and number of these items. For convenience, to specify the valuation for a set of items $\Set_1 = \{s_1^1, \ldots s_1^{k_1} \} $, instead of $v_i(\{s_1^1, \ldots, s_1^{k_1}\}),$ we also use $v_i(s_1^1, \ldots, s_1^{k_1}).$ Similarly, while defining the marginal utility of a set $\Set_2 = \{ s_2^1, \ldots, s_2^{k_2} \}$ over $\Set_1$ instead of writing $v_i(\{s_2^1, \ldots, s_2^{k_2}\} \mid \{s_1^1, \ldots, s_1^{k_1}\}),$ we also use $v_i( s_2^1, \ldots, s_2^{k_2} \mid s_1^1, \ldots, s_1^{k_1}).$

\subsection{Analysis}\label{subsec:algSV}
We will prove Theorem \ref{thm:mainsub} using a series of supporting lemmas. We first prove that in Phase $2$, the minimum marginal utility of an item allocated to an agent over her current allocation from previous iterations of Phase $2$ is not too small. This is the main result that allows us to bound the minimum valuation of the set of items allocated in Phase $2$.  

In the $t^{th}$ iteration of Phase $2$, $\AlgSV$ finds a maximum weight matching. Here, the algorithm tries to assign to each agent an item that gives her the maximum marginal utility over her currently allocated set of items. However, every agent is competing with $n - 1$ other agents to get this item. So, instead of receiving the best item, she might lose a few high ranked items to other agents. Consider the intersection of the set of items that agent $i$ loses to other agents in the $t^{th}$ iteration with the set of items left from her optimal bundle at the beginning of $t^{th}$ iteration.  We will refer to this set of items by $\Set_i^t.$ Let the number of items in $\Set_i^t$ be $k_i^t.$ 

For the analysis of $\AlgSV,$ we also introduce the notion of \textit{attainable items} for every iteration. $\Set_i^t$ is the set of an agent's preferred items that she lost to other agents. The items that are now left are referred as the set of \textit{attainable} items of the agent. Note that in any matching, every agent gets an item equivalent to her best attainable item, that is, an item for which her marginal valuation (over her current allocation) is at least equal to that from her highest marginally valued attainable item. 

For all $i$, we denote the intersection of the set of $\textit{attainable}$ items in the $t^{th}$ iteration and agent $i$'s optimal bundle $\vecx_i^*$ by $\Bar{\vecx}_{i, t}^*$, and let $u_i^* = v_i(\Bar{\vecx}_{i, 1}^*) = v_i(\Bar{\vecx}_i^* \setminus \Set_i^1)$ be the total valuation of \textit{attainable} items at the first iteration of Phase $2$. In the following lemma, we prove a lower bound on the marginal valuation of the set of \textit{attainable} items over the set of items that the algorithm has already allocated to the agent.

\begin{lemma}\label{lem:submod:remaining_val}
 For any $j \in [\tau_i^2 - 1]$, 
\[v_i(\Bar{\vecx}_{i, j+1}^*\mid h_i^1, \dotsc, h_i^j) \geq u_i^* - k_i^2v_i(h_i^1)- \sum_{t=2}^{j}k_i^{t+1} v_i(h_i^t\mid h_i^1, \dotsc, h_i^{t-1})-v_i(h_i^1,h_i^2\dotsc ,h_i^{j})\enspace . \]
\end{lemma}
\begin{proof}
We prove this lemma using induction on the number of iterations $t.$ Consider the base case when $t = 2.$ Agent $i$ has already been allocated $h_i^1.$ She now has at most $\Bar{\tau}_i^*-k_i^1$ items left from $\Bar{\vecx}_i^*$ that are not yet allocated. In the next iteration the agent loses $k_i^2$ items to other agents and receives $h_i^2$. Each of the remaining $\Bar{\tau}_i^* - k_i^1$ items have marginal utility at most $v_i(h_i^1)$ over $\emptyset.$ Thus, the marginal utility of these items over $h_i^1$ is also at most $v_i(h_i^1).$ We bound the total marginal valuation of $\Bar{\vecx}_{i,2}^*$ over $\{h_i^1\}$, by considering two cases. 
 
 \noindent
 \textbf{Case 1:} $h_i^1 \notin \Bar{\vecx}_{i, 1}^*$: By monotonicity of $v_i$, $v_i(\Bar{\vecx}_{i, 2}^* \mid h_i^1) \geq v_i(\Bar{\vecx}_{i, 2}^*) - v_i(h_i^1)= v_i(\Bar{\vecx}_{i, 1}^* \setminus \Set_i^2) - v_i(h_i^1).$
 \medskip 
 
\noindent \textbf{Case 2:} $h_i^1 \in \Bar{\vecx}_{i, 1}^*$: Here, $v_i(\Bar{\vecx}_{i, 2}^* \mid h_i^1) = v_i(\Bar{\vecx}_{i, 2}^* \cup \{ h_i^1\}) - v_i(h_i^1)= v_i(\Bar{\vecx}_{i, 1}^* \setminus \Set_i^2) - v_i(h_i^1).$ 

\noindent
 In both cases, submodularity of valuations and the fact that for all $j \in \Set_i^2, v_i(j) \leq v_i(h_i^1)$ implies,
 \begin{align*}
v_i(\Bar{\vecx}_{i, 2}^* \mid h_i^1) \ \geq\ v_i(\Bar{\vecx}_{i, 1}^*) - v_i(\Set_i^2) - v_i(h_i^1)\ \geq\ u_i^* - k_i^2v_i(h_i^1) - v_i(h_i^1),
 \end{align*}
 
 \noindent
 proving the base case. Now assume the lemma is true for all $t\le r$ iterations, for some $r$, i.e.,
\begin{equation*}
     v_i(\Bar{\vecx}_{i, r}^* \mid h_i^1, \dotsc, h_i^{r - 1}) \geq u_i^* - k_i^2v_i(h_i^1)- \sum_{t=2}^{r - 1}k_i^{t+1} v_i(h_i^t \mid h_i^1, \dotsc, h_i^{t-1})-v_i(h_i^1,h_i^2\dotsc ,h_i^{r-1}).
 \end{equation*}
 Consider the $(r+1)^{st}$ iteration. Again, we analyze two cases. 
 \begin{case}
     \item[$h_i^r \notin \Bar{\vecx}_{i, r}^*$: ]
    \vspace{-5mm}
    \begin{align*}
      v_i(\Bar{\vecx}_{i, r+1}^* \mid h_i^1, \ldots, h_i^r) &= v_i(\Bar{\vecx}_{i, r}^* \setminus \Set_i^{r + 1} \mid h_i^1, \ldots, h_i^r)\\
     &\geq v_i(\Bar{\vecx}_{i, r}^* \mid h_i^1, \ldots, h_i^r) - v_i(\Set_i^{r+1} \mid h_i^1, \ldots, h_i^r) \\
     &\geq v_i(\Bar{\vecx}_{i, r}^* \mid h_i^1, \ldots, h_i^r) - v_i(\Set_i^{r+1} \mid h_i^1, \ldots, h_i^{r - 1})\\
      &\geq v_i(\Bar{\vecx}_{i, r}^* \mid h_i^1, \ldots, h_i^{r - 1}) - v_i(h_i^{r} \mid h_i^1, \ldots, h_i^{r - 1})- v_i(\Set_i^{r+1} \mid h_i^1, \ldots, h_i^{r - 1}).
      \end{align*}
      The submodularity of $v_i$ gives the first two inequalities, and monotonicity of $v_i$ implies the last.
  \item[$h_i^r \in \Bar{\vecx}_{i, r}^*$:]
 \vspace{-5mm}
 \begin{align*}
      v_i(\Bar{\vecx}_{i, r+1}^* \mid h_i^1, \ldots, h_i^r)&= v_i(\Bar{\vecx}_{i, r+1}^* \cup \{ h_i^r \} \mid h_i^1, \ldots, h_i^{r - 1})- v_i(h_i^r \mid h_i^1, \ldots, h_i^{r - 1})\\
     &= v_i(\Bar{\vecx}_{i, r}^* \setminus \Set_i^{r+1} \mid h_i^1, \ldots, h_i^{r - 1}) - v_i(h_i^r \mid h_i^1, \ldots, h_i^{r - 1})\\
    &\geq v_i(\Bar{\vecx}_{i, r}^*  \mid h_i^1, \ldots, h_i^{r - 1}) - v_i(h_i^r \mid h_i^1, \ldots, h_i^{r - 1})- v_i(\Set_i^{r+1} \mid h_i^1, \ldots, h_i^{r-1}). 
 \end{align*}
Here, the second expression follows as $\Bar{\vecx}_{i, r}^* = \Bar{\vecx}_{i, r+1}^* \cup \{h_i^r\} \cup \Set_i^{r+1}$, and the last follows from the definition of submodularity of the valuations.
\end{case}
In both cases, from the induction hypothesis we get,
\begin{equation*}
     \begin{split}
         v_i(\Bar{\vecx}_{i, r+1}^* \mid h_i^1, \ldots, h_i^r) \geq  u_i^* &- k_i^2v_i(h_i^1)- \sum_{t=2}^{r - 1}k_i^{t+1} v_i(h_i^t \mid h_i^1, \dotsc, h_i^{t-1}) -  v_i(h_i^1,h_i^2\dotsc ,h_i^{r-1}) \\ & \quad -v_i(h_i^{r} \mid h_i^1, \ldots, h_i^{r - 1})- v_i(\Set_i^{r+1} \mid h_i^1, \ldots, h_i^{r - 1}).
     \end{split}
     \end{equation*}
     Finally, since $\AlgSV$ assigns the item with highest marginal utility from the set of \textit{attainable} items, and each item in $\Set_i^{r+1}$ is attainable at $r^{th}$ iteration,
     \begin{equation*}
     \begin{split}
     v_i(\Bar{\vecx}_{i, r+1}^* \mid h_i^1, \ldots, h_i^r) \geq\, u_i^* &- k_i^2v_i(h_i^1)- \sum_{t=2}^{r - 1}k_i^{t+1} v_i(h_i^t \mid h_i^1, \dotsc, h_i^{t-1})-v_i(h_i^1,h_i^2\dotsc ,h_i^{r-1}) \\&- v_i(h_i^{r} \mid h_i^1, \ldots, h_i^{r - 1}) -k_i^{r+1}v_i(h_i^r \mid h_i^1, \ldots, h_i^{r - 1})\\
     = \,u_i^* &- k_i^2v_i(h_i^1)-\sum_{t=2}^{r}k_i^{t+1} v_i(h_i^t \mid h_i^1, \dotsc, h_i^{t-1})- v_i(h_i^1,h_i^2\dotsc ,h_i^{r}).\qedhere
     \end{split}
     \end{equation*}
 \end{proof}

The above lemma directly allows us to give a lower bound on the marginal valuation of item received by the agent in $(j+1)^{th}$ iteration over the items received in previous iterations. We state and prove this in the following corollary.
\begin{corollary}\label{corr:submod:oneitemval}
 For any $j \in [\tau_i^2 - 1], $
 \begin{align*}
     v_i(h_i^{j + 1} \mid h_i^1, \dotsc, h_i^j) \geq \frac{1}{\Bar{\tau}_i^* - \sum_{t = 1}^{j + 1} k_i^t} \bigg( u_i^* -k_i^2v_i(h_i^1)- \sum_{t=2}^{j}k_i^{t+1} v_i(h_i^t \mid h_i^1, \dotsc, h_i^{t-1}) -v_i(h_i^1,\dotsc ,h_i^{j})\bigg).
\end{align*}
\end{corollary}
\begin{proof}
In any setting with a set of items $\Set = \{ s_1, \ldots s_k \}$, and a monotone submodular valuation $\mathit{v}$ on this set, if $v(\Set) = \mathit{u}$, then there exists an item $s \in \Set$ such that $\mathit{v}(s) \geq \mathit{u} / k.$ 
Thus, with $\Set=\Bar{\vecx}_{i,j+1}^*$, $k=\Bar{\tau}_i^* - \sum_{t = 1}^{j + 1} k_i^t$, for the submodular valuation function $v_i(\cdot\mid \{h_i^1, \ldots, h_i^j\})$, we can say that at iteration $j+1$, $h_i^{j+1}$ will have a marginal valuation at least,
\begin{equation*}
\frac{1}{\Bar{\tau}_i^* - \sum_{t = 1}^{j + 1} k_i^t} v_i(\Bar{\vecx}_{i, j+1}^* \mid h_i^1, \ldots, h_i^j)\enspace .  
\end{equation*}
Together with Lemma \ref{lem:submod:remaining_val}, this proves the corollary. Note that at any iteration $t$, if the received item $h_i^t$ is from $\Bar{\vecx}_{i, t}^*,$ then the denominator reduces further by $1$, and the bound still holds.
\end{proof}

In the following lemma, we give a lower bound on the total valuation of the items received by the agent in Phase $2$.
 \begin{lemma}\label{lem:main}
     $v_i(h_i^1, \dotsc, h_i^{\tau_i^2}) \ge \frac{u_i^*}{n}.$ 
 \end{lemma}
 \begin{proof}
 Recall that $u_i^*$ is the valuation of the items from $\Bar{\vecx}_i^*$ \emph{after} she loses items in $\Set_i^1$ to other agents in the first iteration of Phase $2$ and $\Bar{\tau}_i^*$ is the number of items in $\Bar{\vecx}_i^*.$ From Corollary \ref{corr:submod:oneitemval}, total valuation of the items obtained by agent $i$ in Phase $2$ is bounded as follows.
 \begin{align*}
      v_i(h_i^1, \dotsc h_i^{\tau_i^2}) & = v_i(h_i^1, \dotsc, h_i^{\tau_i^2 - 1}) + v_i(h_i^{\tau_i^2}\mid h_i^1,\dotsc, h_i^{\tau_i^2-1}) \\
      & \geq v_i(h_i^1, \dotsc, h_i^{\tau_i^2 - 1}) +\frac{1}{\Bar{\tau_i}^*-\sum_{t=0}^{\tau_i^2-1} k_i^{t+1}}\bigg( u_i^* -k_i^2 v_i(h_i^1) 
        -\sum_{t=2}^{\tau_i^2-1}k_i^{t+1} v_i(h_i^t \mid h_i^1,\dotsc h_i^{t-1}) \\ 
	   &  \ \hspace{10cm}  - v_i(h_i^1,\dotsc,h_i^{\tau_i^2-1}) \bigg).
\end{align*}

By definition, $\tau_i^2$ is the last iteration of Phase $2$ in which agent $i$ gets matched to some item. After this iteration, at most $n$ items from her optimal bundle remain unallocated, else she would have received one more item in the $(\tau_i^2+1)^{st}$ iteration. This means the optimal number of items $\Bar{\tau}_i^* - \sum_{t=0}^{\tau_i^2-1} k_i^{t+1} \leq n$, hence the denominator of the second term in the above equation is at most $n$. Again, we note here that if at any iteration $t,$ the item assigned to agent $i$ was from $\Bar{\vecx}_{i, t}^*$, then the denominator will be further reduced by $1$ for all such iterations, and the inequality still remains true when $k_i^t$ is replaced by $k_i^t + 1$. Combined with the fact that an agent can lose at most $n-1$ items in every iteration, we get $k_i^t\le n-1$, implying,
\begin{align*}
      v_i(h_i^1, \dotsc h_i^{\tau_i^2}) &\geq v_i(h_i^1, \dotsc, h_i^{\tau_i^2 - 1}) +  \frac{1}{n}\bigg(u_i^*-k_i^2v_i(h_i^1)-\sum_{t=2}^{\tau_i^2-1}k_i^{t+1} v_i(h_i^t \mid h_i^1,\dotsc h_i^{t-1}) \\
      &\hspace{10cm}- v_i(h_i^1,h_i^2\dotsc,h_i^{\tau_i^2-1})\bigg)\\
      &\geq v_i(h_i^1, \dotsc, h_i^{\tau_i^2 - 1})+ \frac{1}{n}\bigg(u_i^*-(n-1)v_i(h_i^1)- \sum_{t=2}^{\tau_i^2-1}(n-1) v_i(h_i^t \mid h_i^1,\dotsc h_i^{t-1}) \\
      &\hspace{10cm}- v_i(h_i^1,h_i^2\dotsc,h_i^{\tau_i^2-1})\bigg)\\
      &=v_i(h_i^1, \dotsc, h_i^{\tau_i^2 - 1}) + \frac{1}{n}\bigg(u_i^*-
      (n-1)v_i(h_i^1,h_i^2\dotsc,h_i^{\tau_i^2-1}) - v_i(h_i^1,h_i^2\dotsc,h_i^{\tau_i^2-1})\bigg)\\ &=\frac{u_i^*}{n}.\qedhere
\end{align*}
\end{proof}

\begin{remark}
In Lemma $\ref{lem:submod:remaining_val}$ and its subsequent Corollary \ref{corr:submod:oneitemval} and Lemma \ref{lem:main}, if $u_i^* - k_i^2v_i(h_i^1)-\sum_{t=2}^{j}k_i^{t+1} v_i(h_i^t \mid h_i^1, \dotsc, h_i^{t-1})-v_i(h_i^1, \dotsc ,h_i^{j})$ becomes negative for any $j \in [\tau_i^2 - 1]$, then we have
\begin{align*}
   u_i^* &\leq k_i^2v_i(h_i^1)+\sum_{t=2}^{j}k_i^{t+1} v_i(h_i^t \mid h_i^1, \dotsc, h_i^{t-1})+v_i(h_i^1, \dotsc ,h_i^{j}) \\
   &\leq  (n - 1)v_i(h_i^1)+\sum_{t=2}^{j}(n - 1) v_i(h_i^t \mid h_i^1, \dotsc, h_i^{t-1})+ v_i(h_i^1, \dotsc ,h_i^{j}) \\
   &= n \cdot v_i(h_1, \ldots, h_i^j) \leq n \cdot v_i(h_1, \ldots, h_i^{\tau_i^2 - 1}),
\end{align*}
which implies that Lemma $\ref{lem:main}$ holds.
\end{remark}

We now bound the minimum valuation that can be obtained by every agent in Phase $3$. Recall that $g^1_i$ is the item that gives the highest marginal utility over the empty set to agent $i$. Before proceeding, we define 
\[\Goods_i^1 := \{g \in \Goods \ |\ v_i(g \mid \emptyset) \geq v_i(g_i^1 \mid \emptyset)\}\enspace .\]
 \begin{lemma}\label{lem:numwinnerssub}
 Consider the complete bipartite graph where the set of agents $\A$, and the set of items allocated in the first Phase of $\AlgSV$ are the parts, and edge weights are the weighted logarithm of the agent's valuation for the bundle of items containing the item adjacent to the edge and items allocated in Phase $2$. That is, consider  $\Gamma(\A,\Goods=\bigcup_i\vecx_i^1,\W=\{w(i,j)=\eta_i\log(v_i(\{j\}\cup \vecx_i^2))\})$. In this graph, there exists a matching where each agent $i$ gets matched to an item from their highest valued set of items $\Goods_i^1$.
\end{lemma}
\begin{proof}
Among all feasible matchings between the set of agents and the set of items, say $T$, allocated after $t$ iterations of Phase $1$, consider the set of matchings $\M$ where each agent $i$ whose $\Goods_i^1\subseteq T$ is matched to some item in $\Goods_i^1$. Arbitrarily pick a matching from $\M$ where maximum number of agents are matched to an item from their $\Goods_i^1$s. Denote this matching by $\M^t$. Since $|\bigcup_{i\in \Set}\Goods_i^1|\ge |\Set|$ for every set $S$ of agents, in $\M^t$, each agent $i$, who is not matched to an item from their $\Goods_i^1$, has at least one item of $\Goods_i^1$ still unallocated after $t$ iterations.  

Let $\A_t$ denote the set of agents that are not matched to any item from their $\Goods_i^1$ in $\M^t$. We prove by induction on $t$ that $|\A_t|\le n/2^t$. 

For the base case, when $t=1$, we count the number of agents who did not receive any item from their own $\Goods_i^1$ in the maximum weight matching of the algorithm. We know that before the first iteration, every item is unallocated. An agent will not receive any item from $\Goods_i^1$ only if all items from this set are allocated to other agents in the matching. Hence, if $\alpha$ agents did not receive any item from their $\Goods_i^1$, all items from at least $\alpha$ number of $\Goods_i^1$ sets got matched to some agent(s) in the first matching. If $\alpha<n/2$, then more than $n/2$ agents themselves received some item from their $\Goods_i^1$. If $\alpha\ge n/2$, then at least $\alpha$ items, each from a different $\Goods_i^1$ were allocated. In either case, releasing the allocation of the first matching releases at least $n/2$ items, each belonging in a distinct agent's $\Goods_i^1$. Hence, in $\M^1$ at least $n/2$ agents receive an item from their $\Goods_i^1$, and $|\A_1|\le n/2$.

For the inductive step, we assume the claim is true for the first $t$ iterations. That is, for every $k\le t$, in $\M^k$, at most $n/2^k$ agents do not receive an item from their $\Goods_i^1$'s.

Before the $(t+1)^{st}$ iteration begins, we know that for every agent $i$ in $\A_{t}$, at least one item from their $\Goods_i^1$ is still unallocated. Again by the reasoning of the base case, at least half of the agents in $\A_{t}$ will have some item from their $\Goods_i^1$ allocated in the $(t+1)^{st}$ matching (possibly to some other agent). Hence, in $\M^{(t+1)}$, $|\A_{(t+1)}|\le |\A_t|/2$. By the inductive hypothesis, $|\A_{(t+1)}|\le n/2^{(t+1)}.$    \qedhere
 \end{proof}

 \begin{proof}[\textit{Proof of Theorem \ref{thm:mainsub}}]
 From Lemma \ref{lem:main}, 
 \begin{equation*}
     v_i(h_i^1, \dotsc, h_i^{\tau_i^2}) \geq \frac{u_i^*}{n}.
 \end{equation*}

 By Lemma \ref{lem:numwinnerssub}, giving each agent her own $g_i^1$ or some item, denoted by say $h_i^{1*}$, that gives her a marginal utility over $\emptyset$ at least as much as $v_i(g_i^1)$ is a feasible matching before Phase $3$ begins. Therefore, we get,
 \begin{equation}\label{eqn:lower_items}
        {\NSW}(\vecx) \ge \left( \prod\limits_{i = 1}^{n} (v_i(h_i^{1*}, h_i^2, \dotsc, h_i^{\tau_i^2}))^{\eta_i} \right)^{1/(\sum_{i = 1}^n \eta_i)}.
\end{equation}
Since the valuation functions are monotonic,
\begin{align*}
    v_i(h_i^{1*}, h_i^2, \dotsc, h_i^{\tau_i^2})\ \geq\ v_i(h_i^{1*})\ \geq\ v_i(g_i^1)\enspace .
\end{align*}
Phase $1$ of the algorithm runs for $\lceil\log n\rceil $ iterations and each iteration allocates $n$ items. Thus, $| \vecx_i^* \setminus \Bar{\vecx}_i^* | \leq n \lceil\log n\rceil$ and $| \Set_i^1 | \leq n$ implying, $| (\vecx_i^* \setminus \Bar{\vecx}_i^*) \cup \Set_i^1 | \leq n(\log n + 2) $. Thus,
\begin{equation*}
     v_i(g_i^1)\ \geq\ \frac{1}{n(\log n + 2)}v_i((\vecx_i^* \setminus \Bar{\vecx}_i^*) \cup \Set_i^1 )\enspace .
\end{equation*}
Also,
\begin{align*}
    v_i(h_i^{1*}, h_i^1, \dotsc, h_i^{\tau_i^2})\ \geq\ v_i(h_i^1, \dotsc, h_i^{\tau_i^2})\ \geq\ \frac{u_i^*}{n}\ =\ \frac{1}{n} v_i(\Bar{\vecx}_i^* \setminus \Set_i^1)\enspace .
\end{align*}
Thus,
\begin{align*}
    v_i(h_i^{1*}, h_i^1, \dotsc, h_i^{\tau_i^2}) & \geq \frac{1}{2} \left( \frac{1}{n(\log n + 2)}v_i((\vecx_i^* \setminus \Bar{\vecx}_i^*) \cup \Set_i^1) + \frac{1}{n}v_i(\Bar{\vecx}_i^* \setminus \Set_i^1) \right) \\
    &\geq \frac{1}{2} \frac{1}{n(\log n + 2)} v_i( ((\vecx_i^* \setminus \Bar{\vecx}_i^*) \cup \Set_i^1) \cup (\Bar{\vecx}_i^* \setminus \Set_i^1) )\\
    &= \frac{1}{2} \frac{1}{n(\log n + 2)} v_i(\vecx_i^*)\enspace .
\end{align*}
The second inequality follows from the submodularity of valuations. The last bound, together with \eqref{eqn:lower_items} gives,
\begin{align*}
      \NSW(\vecx) &\ \ge\ \left( \prod_{i = 1}^n \left( \frac{1}{2} \frac{1}{n(\log n + 2)}v_i(\vecx_i^*)\right)^{\eta_i} \right)^{1/\sum_{i}\eta_i}\ \ge \ \frac{1}{2} \frac{1}{n(\log n + 2)} \OPT\enspace .\qedhere
\end{align*}
 \end{proof}

\begin{remark} 
We remark that even if Phases $1$ and $2$ perform some kind of repeated matchings, the edge weight definitions make them different. In the proof of Lemma \ref{lem:numwinnerssub}, we require that a maximum weight matching matches agents to items according to agent valuations for the single item. That is, in all iterations of Phase $1$, the edge weights of the graph in the future are the valuation of the agent for the set containing the single item, and not the increase in the agent's valuation upon adding this item to her current allocation. These quantities are different when the valuations are submodular. For lower bounding the valuation from the lower ranked items, we need to consider the marginal increase, as defined in Phase $2$. However, Lemma \ref{lem:numwinnerssub} may not hold true if marginal increase in valuations is considered for the initial iterations, hence Phase $1$ is required. 
\end{remark} 

\section{Hardness of Approximation}\label{sec:hardness}

We complement our results for the submodular case with a $\frac{\mathrm{e}}{(\mathrm{e}-1)}$-factor hardness of approximation. Formally, we prove the following theorem.

\begin{theorem}\label{thm:submod_hardness}
Unless $\classP = \classNP$, there is no polynomial time $\frac{\mathrm{e}}{(\mathrm{e}-1)}$-factor approximation algorithm for the submodular $\NSW$ problem, even when agents are symmetric and have identical valuations. 
\end{theorem}
\begin{proof}
We show this using the hardness of approximation result of the $\Alloc$ problem proved in \cite{KhotLMM08}. We first summarize the relevant parts of \cite{KhotLMM08}. The $\Alloc$ problem is to find an allocation of a set of indivisible items among a set of agents with monotone submodular utilities for the items, such that the sum of the utilities of all agents is maximized. Note that if the valuation functions were additive, the problem is trivial, and an optimal allocation gives every item to the agent who values it the most. To obtain a hardness of approximation result for the submodular case, the $\Maxcolor$ problem is reduced to the $\Alloc$ problem. $\Maxcolor$, the problem of determining what fraction of edges of a graph can be properly colored when $3$ colors are used to colors all vertices of the graph, is known to be $\classNP$-Hard to approximate within some constant factor $c$. The reduction from $\Maxcolor$ generates an instance of $\Alloc$ with symmetric agents having identical submodular valuation functions for the items. The reduction is such that for instances of $\Maxcolor$ with optimal value $1$, the corresponding $\Alloc$ instance has an optimal value of $nV$, where $n$ is the number of agents in the instance, and $V$ is a function of the input parameters of $\Maxcolor$. In this case, every agent receives a set of items of utility $V$. For instances of $\Maxcolor$ with optimal value at most $c$, it is shown that the optimal sum of utilities of the resulting $\Alloc$ instance cannot be higher than $(1-1/\mathrm{e})nV$. 

For proving hardness of submodular $\NSW$ problem, observe that the input of the $\Alloc$ and $\NSW$ problems are the same. So, we consider the instance generated by the reduction as that of an $\NSW$ maximizing problem. From the results of \cite{KhotLMM08}, we can prove the following claims.
\begin{itemize}
    \item If the optimal value of $\Maxcolor$ is $1$, then the $\NSW$ of the reduced instance is $V$. As every agent receives a set of items of value $V$, the $\NSW$ is also $V$.
    \item If the optimal value of $\Maxcolor$ is at most $c$, then the $\NSW$ is at most $(1-1/\mathrm{e})V$. Applying the AM-GM inequality establishes that the $\NSW$ is at most $1/n$ times the sum of utilities, which is proven to be at most $(1-1/\mathrm{e})nV$. 
\end{itemize}
As $\Maxcolor$ cannot be approximated within a factor $c$, thus $\NSW$ of a problem with submodular utilities cannot be approximated within a factor $\frac{\mathrm{e}}{(\mathrm{e}-1)}$. 

As the $\Alloc$ problem now considered as an $\NSW$ problem had symmetric agents and identical submodular valuation functions, the $\NSW$ problem also satisfies these properties.
\end{proof}

\section{Special Cases}\label{sec:special_cases}

\subsection{Submodular $\NSW$ with Constant Number of Agents}
In this section, we describe a constant factor algorithm for a special case of the submodular $\NSW$ problem. Specifically, we prove the following theorem. 
\begin{theorem}\label{thm:submod_constant}
For any constant $\epsilon>0$ and a constant number of agents $n\geq 2$, there is a $(1-1/e-\epsilon)$-factor approximation algorithm for the $\NSW$ problem with monotone submodular valuations, in the value oracle model. Additionally, this is the best possible factor independent of $n$, and any factor better than $(1-(1-1/n)^n+\epsilon)$ would require exponentially many queries, unless $\classP=\classNP$.
\end{theorem}

The key results that establish this result are from the theory of submodular function maximization developed in \cite{ChekuriVZ10}. The broad approach for approximately maximizing a discrete monotone submodular function is to optimize a popular continuous relaxation of the same, called the multilinear extension, and round the solution using a randomized rounding scheme. We will use an algorithm that approximately maximizes multiple discrete submodular functions, described in \cite{ChekuriVZ10}, as the main subroutine of our algorithm for the submodular $\NSW$ problem, hence first we give an overview of it, starting with a definition of the multilinear extension.

\begin{definition}[Multilinear Extension of a submodular function] Given a discrete submodular function $f:2^m\rightarrow \mathbb{R_+}$, its multilinear extension $F:[0,1]^m\rightarrow \mathbb{R_+}$, at a point $y\in [0,1]^m$, is defined as the expected value of $f(z)$ at a point $z\in \{0,1\}^m$ obtained by rounding $y$ such that each coordinate $y_i$ is rounded to $1$ with probability $y_i$, and to $0$ otherwise. That is, 
$$F(y)=\mathbb{E}[f(z)]=\sum\limits_{X\subseteq [m]}f(X)\prod\limits_{i\in X}y_i\prod\limits_{i\notin X}(1-y_i).$$
\end{definition}

The following theorem proves that the multilinear extensions of multiple discrete submodular functions defined over a matroid polytope can be simultaneously approximated to optimal values within constant factors.
\begin{theorem}\cite{ChekuriVZ10}\label{thm:multiple_submod}
Consider monotone submodular functions $f_1,\dotsc,f_n : 2^N \rightarrow \mathbb{R_+}$, their multilinear extensions $F_i : [0,1]^N \rightarrow \mathbb{R_+}$ and a matroid polytope $P \subseteq [0,1]^N$. There is a polynomial time algorithm which, given $V_1,...,V_n \in \mathbb{R_+}$, either finds a point $x \in P$ such that $F_i(x) \geq  (1-1/e)V_i$ for each $i$,  or returns a certificate that there is no point $x \in P$ such that $F_i(x) \geq V_i$ for all $i$. 
\end{theorem}

Given a discrete monotone submodular function $f$ defined over a matroid, a rounding scheme called the \textit{swap rounding} algorithm can be applied to round a solution of its multilinear extension to a feasible point in the domain of $f$, which is an independent set of the matroid. At a high level, in the rounding scheme, it is first shown that every solution of the multilinear extension can be expressed as a convex combination of independent sets such that for any two sets $S_0$ and $S_1$ in the convex combination, there is at least one element in each set that is not present in the other, that is $\exists e_0\in S_0\backslash S_1$ and $\exists e_1\in S_1\backslash S_0$ . The rounding method then iteratively merges two arbitrarily chosen sets $S_0$ and $S_1$ into one new set as follows. Until both sets are not the same, one set $S_i$ is randomly chosen with probability proportional to the coefficient of its original version in the convex combination $\beta_i$, that is $S_i$ is chosen with probability $\beta_i/(\beta_0+\beta_1)$, and altered by removing $e_i$ from it and adding $e_{1-i}$. The coefficient of the new set obtained by this merge process is the sum of those of the sets merged, i.e., $\beta_0+\beta_1$. 

The following lower tail bound proves that with high probability, the loss in the function value by swap rounding is not too much.
\begin{theorem}\cite{ChekuriVZ10}\label{thm:tail_bound_submod}
 Let $f:\{0,1\}^n \rightarrow\mathbb{R_+}$ be a monotone submodular function with marginal values in $[0,1]$, and $F:[0,1]^n \rightarrow \mathbb{R_+}$ its multilinear extension. Let $(x_1,...,x_n) \in P(M)$ be a point in a matroid polytope and $(X_1,...,X_n) \in \{0,1\}^n$ a random solution obtained from it by randomized swap rounding. Let $\mu_0 = F(x_1,...,x_n)$ and $\delta > 0$. Then 
 $$Pr[f(X_1,...,X_n) \leq (1-\delta)\mu_0] \leq e^{-\mu_0\delta^2/8}.$$
\end{theorem}

In short, for a matroid $M(X,I)$, given monotone submodular functions $f_i:\{0,1\}^m\rightarrow \mathbb{R_+}, i\in [n]$ over the matroid polytope, and values $v_i, i\in [n]$, there is an efficient algorithm that determines if there is an independent set $S\in I$ such that $f_i(S)\geq (1-1/e)v_i$ for every $i$.

To use this algorithm to solve the submodular $\NSW$ problem, we define a matroid $M(X,I)$ as follows. This construction was first described in \cite{LehmannLN06}, and also used for approximating the submodular welfare in \cite{Vondrak08}. From the sets of agents $\A$ and items $\Goods$, we define the ground set $X=\A\times \Goods$. The independent sets are all feasible integral allocations $I=\{S\subseteq X\mid \forall j: |S\cap \{\A\times \{j\}\}|\leq 1 \}$. The valuation functions of every agent $u_i:\{0,1\}^m\rightarrow \mathbb{R_+}$ translate naturally to submodular functions over this matroid $f_i:I\rightarrow \mathbb{R_+}$, with $f_i(S)=u_i(\Goods_i)$, where $\Goods_i=\{j\in \Goods\mid (i,j)\in S\}$.
With this construction, for any set of values $V_i, i\in[n]$, checking if there is an integral allocation of items that gives valuations at least (approximately) $V_i$ to each agent $i$ is equivalent to checking if there is an independent set in this matroid that has value $V_i$ for every agent $i$. 

\begin{algorithm}[t]
\DontPrintSemicolon
\SetAlgoLined
\SetKwInOut{Input}{Input}
\SetKwInOut{Output}{Output}
\Input{A set $\A$ of $n$ agents with weights $\eta_i,\  \forall i\in\A$, a set $\Goods$ of $m$ indivisible items, and monotone submodular valuations $u_i:2^{\Goods}\rightarrow \mathbb{R_+}$.}  
\Output{An allocation that approximates the $\NSW$.}\medskip
 For any value $Max>0$ that is a power of $(1+\delta)$, scale all valuation functions such that $u_i(\Goods)=Max$ for all $i$. \tcp*{$Max$ is an upper bound on $\NSW$ objective}
 $OPT=Max$ \tcp*{$OPT$ is the optimal $\NSW$ objective}
 define $\beta >0,\ \delta >0$ as small positive constants\;
\While{$OPT \leq Max$}{                    
    flag=0\;
    \For{any set in $V=\{[V_1,V_2,\dotsc,V_n]\mid \prod_i V_i=OPT, \forall i: V_i=(1+\delta)^{k_i}\ for\ some\ k_i\}$ }{
         \If{there is an allocation $\mathbf x$ of $\Goods$ such that $u_i(\vecx_i)\geq (1-1/e)V_i$ for all $i$}{
                     ${\mathbf x}^*=\mathbf x, flag=1$ \tcp*{$flag=1$ if current $OPT$ value is feasible}}
    }
    \eIf{flag=1}{
        $OPT =OPT+(Max+\beta-OPT)/2$ \tcp*{search if a higher value is also feasible. Adding $\beta$ ensures $OPT>Max$ finally, and algorithm converges}}{
        $Max=OPT$, $OPT=OPT/2$ \tcp*{search for a lower feasible value}}
        $OPT=$ nearest power of $(1+\delta)$ greater than $OPT$\;
        $Max=$ nearest power of $(1+\delta)$ greater than $Max$
}
\Return{ ${\mathbf x}^*$ }
\caption{Approximate the Submodular $\NSW$ with constant number of agents}\label{Alg:submod_constant}
\end{algorithm}

The algorithm for approximating the $\NSW$ is now straightforward, and given in Algorithm \ref{Alg:submod_constant}. Essentially, we guess the optimal $\NSW$ value $OPT$, and the utility of every agent in the optimal allocation $V_i$, and check if there is an allocation $X$ that gives every agent $i$ a bundle of value at least (approximately) $V_i$. As every agent can receive at most $Max$ utility, $Max$ is a trivial upper bound for the maximum value of $\NSW$, hence we perform a binary search for the optimal value in the range $(0,Max]$. Searching for sets $V_i$ by enumerating only those sets with values that are powers of $(1+\delta)$ for some constant $\delta>0$ will reduce the time complexity of the algorithm to $O(poly(\log(Max)/\delta))$ instead of $O(poly(Max))$, by changing the approximation factor to $(1-1/e)(1-\delta)\leq (1-1/e-\epsilon)$ for some $\epsilon>0$.  

The hardness claim in Theorem \ref{thm:submod_constant} follows from the proof of Theorem \ref{thm:submod_hardness}. It was shown that in the case where the optimal value of the $\Maxcolor$ instance was $1$, every agent in the reduced $\NSW$ instance received a bundle of items of value $V$, else the total $\NSW$ could not be more than $(1-(1-1/n)^n)V$.

\subsection{Symmetric Additive $\NSW$}
We now prove that $\AlgAdd$ gives an allocation that also satisfies the $\EF$ property, making it not only approximately efficient but also a fair allocation. $\EF$ is formally defined as follows.

\begin{definition}[\cite{Budish11}]
Envy-Free up to one item ($\EF$): An allocation $\vecx$ of $m$ indivisible items among $n$ agents satisfies the envy-free up to one item property, if for any pair of agents $i,\hat{i}$, either $v_i(\vecx_i) \geq v_i(\vecx_{\hat{i}})$, or there exists some item $g\in \vecx_{\hat{i}}$ such that $v_i(\vecx_i) \geq v_i(\vecx_{\hat{i}}\backslash \{g\})$.  
\end{definition}

That is, if an agent $i$ values another agent $\hat{i}$'s allocation more than her own, which is termed commonly by saying agent $i$ \textit{envies} agent $\hat{i}$, then there must be some item in $\hat{i}$'s allocation upon whose removal this envy is eliminated.

\begin{theorem}\label{thm:additive_ef1}
The output of $\AlgAdd$ satisfies the $\EF$ fairness property. 
\end{theorem}
\begin{proof}
For every agent $i$ and $j\geq 1$, the item $h_j^i$ allocated to $i$ in the $j^{th}$ iteration of $\AlgAdd$ is valued more by $i$ than all items $h^{i'}_{k}$, $k>j$ allocated to any other agent $i'$ in the future iterations, as otherwise $i$ would have been matched to the other higher valued item in the $j^{th}$ matching. Hence, $\sum\limits_{t=1}^j v_i(h_t^i) \geq \sum\limits_{t=2}^j v_i(h_t^{i'})$. That is, after removing the first item $h^{i'}_1$ from any agent's bundle, the sum of valuations of the remaining items for agent $i$ is not higher than her current total valuation. Thus, after removing the item allocated to any agent in the first matching, agent $i$ does not envy the remaining bundle, making the allocation $\EF$. 
\end{proof}

\begin{remark}
We note that the same proof implies that our algorithm satisfies the strong $\EF$ property, defined in~\cite{ConitzerFSV19}. Intuitively, an allocation satisfies the strong $\EF$ property if upon removing the same item from agent $i$ bundle, no other agent $j$ envies $i$, for all $i$ and $j$. Formally, 
\begin{definition}[Strong $\EF$]
An allocation $\vecx$ satisfies strong $\EF$ if for every agent $i\in \A$, there exists an item $g_i\in \vecx_i$ such that no other agent envies the set $\vecx_i \backslash \{g_i\}$, i.e., $\forall j\in \A,\ v_j(\vecx_j)\geq v_j(\vecx_i\backslash\{g_i\})$.
\end{definition}
\end{remark}

\subsection{Symmetric Restricted Additive $\NSW$}
For the special case when the valuations are restricted, meaning the valuation of any item $v_{ij}$ is either some value $v_j$ or $0$, we now prove $\AlgAdd$ gives a constant factor approximation to the optimal $\NSW$.
\begin{theorem}
$\AlgAdd$ solves the symmetric $\NSW$ problem for restricted additive valuations within a factor $1.45$ of the optimal.
\end{theorem}
\begin{proof}
We prove that $x^*$, the allocation returned by $\AlgAdd$, is Pareto Optimal ($\PO$). Combined with the statement of Theorem \ref{thm:additive_ef1}, and a result of \cite{BarmanKV18} which proves that any allocation that satisfies both $\EF$ and $\PO$ approximates $\NSW$ with symmetric, additive valuations within a $1.45$ factor, we get the required result.
An allocation of items $x$ is called Pareto Optimal when there is no other allocation $x'$ where every agent gets at least as much utility as in $x$, and at least one agent gets higher utility. In the restricted valuations case, every item adds valuation either $0$ or $v_j$ to some agent's utility. Thus, the sum of valuations of all agents in any allocation is at most $\sum_j v_j$. Observe that $\AlgAdd$ can easily be modified so that it allocates every item to some agent who has non zero valuation for it. Then, the sum of valuations of all agents in the allocation returned by $\AlgAdd$ is $\sum_j v_j$. No other allocation can give an agent strictly higher utility without decreasing another agent's utility. Hence, $x^*$ is a Pareto Optimal allocation.
\end{proof}

\begin{remark}
We remark that Theorem $\ref{thm:additive_ef1}$ also holds for general additive valuations. However, for the general case, the $\PO$ property does not always hold. Consider for example the case where we have two agents $\{\mathsf{A}, \mathsf{B}\}$ and four items $\{g_1, g_2, g_3, g_4 \}.$ Agent $\mathsf{A}$ values the items at $\{2 + \epsilon, 2, \epsilon, \epsilon\}$ and agent $\mathsf{B}$ values them at $\{1, 1, 1, 1\}.$ $\AlgAdd$ allocates items $g_1, g_3$ to agent $\mathsf{A}$ and items $g_2, g_4$ to agent $\mathsf{B}.$ However, we can swap items $g_2$ and $g_3$ to get an allocation that Pareto dominates the allocation output by the algorithm.
\end{remark}

\section{Tightness of the analysis}\label{sec:tightness}
\subsection{Subadditive Valuations}
The matching approach does not extend to agents with subadditive valuation functions. Here the valuation functions satisfy the subadditivity property:
$$v(\Set_1\cup \Set_2)\le v(\Set_1)+v(\Set_2),$$
for any subsets $\Set_1,\Set_2$ of the set of items $\Goods$. 

A counter example that exhibits the shortcomings of the approach is as follows.
Consider an instance with $2$ agents and $m$ items. Assume $m$ is even. Denote the set of items by $\Goods = \{ g_1, g_2, \dotsc, g_m\}.$  Let $\Goods_1 = \{ g_1, g_2, \dotsc, g_{m/2} \}$ and $\Goods_2 = \{ g_{m/2 + 1} , \dotsc, g_m \}$. The valuation function for agent $i \in \{ 1, 2\}$ is as follows.
\begin{equation*}
    v_i(\Set) = v_i(\Set_1 \cup \Set_2) = \max\{ M, | \Set_i | \cdot M \} \quad \forall \Set \subseteq \Goods, \, \Set_1 \subseteq \Goods_1, \, \Set_2 \subseteq \Goods_2.
\end{equation*}
Note that these valuation functions are subadditive, but not submodular. 

The allocation that maximizes the $\NSW$ allocates $\Goods_1$ to agent $1$ and $\Goods_2$ to agent $2.$ The optimal $\NSW$ is $mM/2.$

Now, $\AlgSV$ may proceed in the following way. Since the marginal utility of each item over $\emptyset$ is $M,$ the algorithm can pick any of the items for either of the agents. Suppose the algorithm gives $g_{m/2 + 1}$ to agent $1$ and $g_1$ to agent $2.$ In the next iteration, for agent $1$ ($2$) the marginal utility of any item over $g_{m/2 + 1}$ $(g_1)$ is $0.$  Thus, again the algorithm is at liberty to allocate any item to either of the agents. Now again the algorithm gives exactly opposite allocation as compared to the optimal allocation and gives agent $1$ item $g_{m/2 + 2}$ and gives agent $2$ item $g_2.$ For each iteration this process repeats and ultimately the bundles allocated by algorithm are exactly opposite of the bundles allocated by optimal. The re-matching step first releases $\lceil\log{n}\rceil$ matchings, or $2$ items from both agent allocations, and re-matches them. This may not change the allocations as both agents have already received their best item. The $\NSW$ of the algorithm's allocation is $(M^2)^{1/2}=M,$ giving an approximation ratio of $\Omega(m)$ with the optimal $\NSW$. Even if we increase the number of agents, the factor cannot be made independent of $m$, the number of items. 

The problem in the subadditive case is the myopic nature of each iteration in $\AlgSV$. In each iteration the algorithm only sees one step ahead. At any of the iterations, had the algorithm been allowed to pick and allocate multiple items instead of $1$, it would have been able to select a subset of items from its correct optimal bundle. 

This problem does not arise in the additive case because the valuation of an item here is independent of other items. Submodular valuations allow a minimum marginal utility over an agent's current allocation for items allocated in future iterations, hence this issue does not arise there too.

\subsection{$\XOS$ Valuations}
The following example shows that $\AlgSV$ does not extend to $\XOS$ valuations either. $\XOS$ is a class of valuation functions that falls between subadditive and submodular valuation functions, defined as follows. A set of additive valuation functions, say $\{ \ell_1, \ldots, \ell_k  \}$, is given, and the $\XOS$ valuation of a set of items $\Set$ is the maximum valuation of this set according to any of these additive valuations. i.e., $v(\Set)=\max_{i\in[k]}\{\ell_i(\Set)\}.$ 

To see why the algorithm does not extend to this class of functions, consider the following counter example. We have $n = 2$ agents and $m = 2k$ items, for some $k>3$. Each agent $i \in \{1, 2\}$ has 2 valuation functions $\ell^i_1, \ell^i_2.$ The following two tables pictorially depict these valuations. Each entry $(\ell_h^i,g_j),\ h\in [2], i\in [2], j\in [2k]$ denotes agent $i$'s valuation according to function $\ell_h^i$ for item $g_j$. \\
For agent 1:
\begin{center}
    \begin{tabular}{c|c|c|c|c|c|c|c|c|c|c|c}
         & $g_1$ & $g_2$  & $\ldots$ & $g_k$ & $g_{k+1}$ & $g_{k+2}$ & $g_{k+3}$ & $g_{k+4}$ & \ldots & $g_{2k}$  \\
         \hline
         $\ell_1^1$ & M & M  & $\ldots$ & M  & 0 & 0 & 0 & 0 & $\ldots$ & 0 \\
         $\ell_2^1$ & 0 & 0  & $\ldots$ & 0 & M + $\epsilon$ & M + $\epsilon$ & M + $\epsilon$ & $\epsilon$ & \ldots & $\epsilon$ \\
    \end{tabular}
\end{center}
For agent 2:
\begin{center}
    \begin{tabular}{c|c|c|c|c|c|c|c|c|c|c|c}
         & $g_1$ & $g_2$ & $g_3$ & $g_4$ & $\ldots$ & $g_k$ & $g_{k+1}$ & $g_{k+2}$ & \ldots & $g_{2k}$  \\
         \hline
         $\ell_1^2$ & 0 & 0 & 0 & 0 & $\ldots$ & 0  & M & M & $\ldots$ & M \\
         $\ell_2^2$ & M+$\epsilon$ & M+$\epsilon$ & M+$\epsilon$ & $\epsilon$ & $\ldots$ & $\epsilon$ & 0 & 0 & \ldots & 0 \\
    \end{tabular}
\end{center}
Here $M$ is any large value, and $\epsilon>0$ is negligible.

The allocation optimizing $\NSW$ clearly allocates the first $k$ items to agent $1$ and the next $k$ items to agent $2$, resulting in the $\NSW$ value $Mk.$ $\AlgSV$ on the other hand allocates items $g_1, g_2$ to agent $2$ and items $g_{k+1}, g_{k+2}$ to agent $1$ in Phase $1$. In Phase $2$, it gives $g_3$ to agent $2$ and $g_{k+3}$ to agent $1.$ After this, for all other iterations of Phase $2,$ items $g_{j}, j \leq k$ have zero marginal utility for agent $1$ and items $g_j, j \geq (k+1)$ have zero marginal utility for agent $2$. Thus, $\AlgSV$ allocates items $g_3 \ldots g_{k}$ to agent $2$ and items $g_{k+3}, \ldots, g_{2k}$ to agent $1$ in Phase $2.$ Phase $3$ reallocates items of Phase $1$ $-$ $g_1, g_2, g_{k+1}, g_{k+2}$ allocating $g_{k+1}, g_{k+2}$ to agent $1$ and $g_1$, $g_2$ to agent $2.$ Thus the $\NSW$ of the allocation given by $\AlgSV$ is $(3M + k \cdot \epsilon)$. Hence, the approximation ratio of $\AlgSV$ cannot be better than $(MK)/(3M +k\epsilon)$ or $\Omega(k)=\Omega(m)$ when the valuation functions are $\XOS$.

\subsection{Asymmetric Additive $\NSW$}\label{subsec:AsymmetricNSWExample} We describe an example to prove the analysis of this case is tight. 
Consider an $\NSW$ instance with $n$ agents, referred by $\{1,2\dotsc,n\}$ and $m$ sets of $n^2$ items, referred by $\Goods=\{\Goods_i\mid i\in [m]\},$ where every $\Goods_i=\{g_{i,1}\dotsc,g_{i,n^2}\}\}$. The first agent has weight $W$, while the remaining agents have weight $1$. The valuation function of agent $1$ is as follows.
\begin{equation*}
    \begin{split}
        v_1(g_{i,j})=\begin{cases}
        M & j\in [n], i\in [m]\\
        0 & otherwise.
        \end{cases}
    \end{split}
\end{equation*}
The remaining agents have valuations for items as follows.
\begin{equation*}
    \begin{split}
        \forall k\in [n], k\ne 1:\ v_k(g_{i,j})=\begin{cases}
        M+\epsilon & j\in [n], i\in [m],\epsilon>0\\
        M+\bar{\epsilon} & (k-1)n+1\le j \le kn, i\in [m], \epsilon>\bar{\epsilon}>0\\
        0 & otherwise.
        \end{cases}
    \end{split}
\end{equation*}
It is easy to verify that the optimal allocation that maximizes $\NSW$ gives all items $g_{i,j}$ for $i$ between $(k-1)n+1$ and $kn$ to agent $k$. That is, the $k^{th}$ agent receives the $k^{th}$ set of $n$ items from each of the $m$ sets of $n^2$ items. $\AlgAdd$ on the other hand allocates items as follows. For the first graph, it computes $u_i$ as $M(m-2)n$ for the first agent, and $(M+\epsilon)(m-2)n+(M+\epsilon')(mn)$ for a small $\epsilon'>0$, for the rest. The first max weight matching then allocates to each agent one item from agent $1$'s optimal bundle.  
The following iterations also err as follows. Until the first $n$ items from every set are allocated, irrespective of the agent weights, every agent receives one item from this set if available. In the remaining matchings, agent $1$ does not receive any item, and the other agents get all items in their optimal bundles. The ratio of the $\NSW$ products of the optimal allocation and the algorithm's allocation is as follows.
\begin{equation*}
    \begin{split}
    \frac{\NSW(\vecx)}{\OPT}&\le \left(\frac{(mM)^W(mn\cdot (M+\bar{\epsilon})+m(M+\epsilon))^{n-1}}{(mn\cdot M)^W(mn\cdot (M+\bar{\epsilon}))^{n-1}}\right)^{1/(W+(n-1))}\\
    &\le \left(\frac{(mM)^W(2mn\cdot M)^{n-1}}{(mn\cdot M)^W(mn\cdot M)^{n-1}}\right)^{1/(W+(n-1))}\\
    &\le 2 \left(\frac{1}{n}\right)^{W/(W+(n-1))}. 
    \end{split}
\end{equation*}
With increasing $W$, asymptotically this ratio approaches $2/n$.

\begin{remark}\label{rem:AsymSym}
It is natural to ask if the asymmetric $\NSW$ problem is harder than the symmetric problem. As $\AlgAdd$ is the first non-trivial algorithm for the asymmetric problem, we would like to find if it gives a better approximation factor when applied to a symmetric agents instance. However, after considerable effort, we could not resolve this question definitively. Like the above example, we could not find an example of a symmetric instance for which our analysis was tight. Our conjecture is that $\AlgAdd$ gives a better factor for the symmetric problem, and that the symmetric case itself is easier than the asymmetric $\NSW$ case.
\end{remark}

\section{Conclusions}\label{sec:discussion}
In this work, we have shown two algorithms $\AlgAdd$ and $\AlgSV$. $\AlgAdd$ approximately maximizes the $\NSW$ for the asymmetric additive case within a factor of $O(n)$, while $\AlgSV$ optimizes the asymmetric submodular $\NSW$ within a factor of $O(n\log n)$. We also completely resolve the submodular $\NSW$ problem for the case when there are a constant number of agents, with an $e/(e-1)$ approximation factor algorithm, and a matching hardness of approximation proof. Our algorithms also satisfy other interesting fairness guarantees for smaller special cases, namely, $\EF$ for the additive valuations case, and a better $1.45$ factor approximation for the symmetric agents with restricted additive valuations case. 

Our work has initiated the investigation of the $\NSW$ problem for general cases, and raises several interesting questions. First, we ask if the approximation factor $O(n)$, given by $\AlgAdd$, is the best possible for the asymmetric additive $\NSW$ problem. While $\AlgAdd$ cannot give better than a linear factor guarantee, as proved by an example in Section \ref{subsec:AsymmetricNSWExample}, there could be another algorithm with a sub-linear approximation factor. 

Another problem is the special case where the agent weights are separated by a constant factor. Approaches known for the symmetric $\NSW$ problem fail to extend even to the highly restricted case when agent weights are either $1$ or $2$. $\AlgAdd$ does not seem to give a better guarantee for this case either, and we ask if this question is easier than the general asymmetric agents case. One way to resolve this query is to find an algorithm with an approximation factor equal to some polynomial function of the ratio of the largest to smallest weight. 

A third direction is to resolve the symmetric $\NSW$ problem for  valuation functions that are more general than additive(-like). For instance, the symmetric agents with $\OXS$ valuations (defined in ~\cite{LehmannLN06}) case has not been explored yet.  
\medskip

\noindent
\textbf{Acknowledgments.} We thank Chandra Chekuri and Kent Quanrud for pointing us to relevant literature in submodular function maximization theory, and having several fruitful discussions about the same.

\bibliographystyle{alpha}{}
\bibliography{literature}

\end{document}